\documentclass[
reprint,
superscriptaddress,
nofootinbib,
aps, pra,longbibliography
]{revtex4-1}
\usepackage{graphics}
\usepackage{custom}
\usepackage{mathtools}
\usepackage{tikz-cd}
\usepackage{multirow}
\usepackage{booktabs}
\usepackage{dsfont}

\usepackage{enumitem}

\newtheorem{Theorem}{Theorem}
\newtheorem{Lemma}{Lemma}

\def\be{\begin{equation}}
\def\ee{\end{equation}}
\def\ba{\begin{eqnarray}}
\def\ea{\end{eqnarray}}

\newcommand\q{\quad}
\def\Nl{{\mathchoice
{\setbox0=\hbox{$\displaystyle\rm N$}\hbox{\hbox to0pt
{\kern0.4\wd0\vrule height0.9\ht0\hss}\box0}}
{\setbox0=\hbox{$\textstyle\rm N$}\hbox{\hbox to0pt
{\kern0.4\wd0\vrule height0.9\ht0\hss}\box0}}
{\setbox0=\hbox{$\scriptstyle\rm N$}\hbox{\hbox to0pt
{\kern0.4\wd0\vrule height0.9\ht0\hss}\box0}}
{\setbox0=\hbox{$\scriptscriptstyle\rm N$}\hbox{\hbox to0pt
{\kern0.4\wd0\vrule height0.9\ht0\hss}\box0}}}}
\def\Zl{{\mathchoice
{\setbox0=\hbox{$\displaystyle\rm Z$}\hbox{\hbox to0pt
{\kern0.4\wd0\vrule height0.9\ht0\hss}\box0}}
{\setbox0=\hbox{$\textstyle\rm Z$}\hbox{\hbox to0pt
{\kern0.4\wd0\vrule height0.9\ht0\hss}\box0}}
{\setbox0=\hbox{$\scriptstyle\rm Z$}\hbox{\hbox to0pt
{\kern0.4\wd0\vrule height0.9\ht0\hss}\box0}}
{\setbox0=\hbox{$\scriptscriptstyle\rm Z$}\hbox{\hbox to0pt
{\kern0.4\wd0\vrule height0.9\ht0\hss}\box0}}}}
\def\Ql{{\mathchoice
{\setbox0=\hbox{$\displaystyle\rm Q$}\hbox{\hbox to0pt
{\kern0.4\wd0\vrule height0.9\ht0\hss}\box0}}
{\setbox0=\hbox{$\textstyle\rm Q$}\hbox{\hbox to0pt
{\kern0.4\wd0\vrule height0.9\ht0\hss}\box0}}
{\setbox0=\hbox{$\scriptstyle\rm Q$}\hbox{\hbox to0pt
{\kern0.4\wd0\vrule height0.9\ht0\hss}\box0}}
{\setbox0=\hbox{$\scriptscriptstyle\rm Q$}\hbox{\hbox to0pt
{\kern0.4\wd0\vrule height0.9\ht0\hss}\box0}}}}
\def\Rl{{\mathchoice
{\setbox0=\hbox{$\displaystyle\rm R$}\hbox{\hbox to0pt
{\kern0.4\wd0\vrule height0.9\ht0\hss}\box0}}
{\setbox0=\hbox{$\textstyle\rm R$}\hbox{\hbox to0pt
{\kern0.4\wd0\vrule height0.9\ht0\hss}\box0}}
{\setbox0=\hbox{$\scriptstyle\rm R$}\hbox{\hbox to0pt
{\kern0.4\wd0\vrule height0.9\ht0\hss}\box0}}
{\setbox0=\hbox{$\scriptscriptstyle\rm R$}\hbox{\hbox to0pt
{\kern0.4\wd0\vrule height0.9\ht0\hss}\box0}}}}
\def\Cl{{\mathchoice
{\setbox0=\hbox{$\displaystyle\rm C$}\hbox{\hbox to0pt
{\kern0.4\wd0\vrule height0.9\ht0\hss}\box0}}
{\setbox0=\hbox{$\textstyle\rm C$}\hbox{\hbox to0pt
{\kern0.4\wd0\vrule height0.9\ht0\hss}\box0}}
{\setbox0=\hbox{$\scriptstyle\rm C$}\hbox{\hbox to0pt
{\kern0.4\wd0\vrule height0.9\ht0\hss}\box0}}
{\setbox0=\hbox{$\scriptscriptstyle\rm C$}\hbox{\hbox to0pt
{\kern0.4\wd0\vrule height0.9\ht0\hss}\box0}}}}
\def\Hl{{\mathchoice
{\setbox0=\hbox{$\displaystyle\rm H$}\hbox{\hbox to0pt
{\kern0.4\wd0\vrule height0.9\ht0\hss}\box0}}
{\setbox0=\hbox{$\textstyle\rm H$}\hbox{\hbox to0pt
{\kern0.4\wd0\vrule height0.9\ht0\hss}\box0}}
{\setbox0=\hbox{$\scriptstyle\rm H$}\hbox{\hbox to0pt
{\kern0.4\wd0\vrule height0.9\ht0\hss}\box0}}
{\setbox0=\hbox{$\scriptscriptstyle\rm H$}\hbox{\hbox to0pt
{\kern0.4\wd0\vrule height0.9\ht0\hss}\box0}}}}
\def\Ol{{\mathchoice
{\setbox0=\hbox{$\displaystyle\rm O$}\hbox{\hbox to0pt
{\kern0.4\wd0\vrule height0.9\ht0\hss}\box0}}
{\setbox0=\hbox{$\textstyle\rm O$}\hbox{\hbox to0pt
{\kern0.4\wd0\vrule height0.9\ht0\hss}\box0}}
{\setbox0=\hbox{$\scriptstyle\rm O$}\hbox{\hbox to0pt
{\kern0.4\wd0\vrule height0.9\ht0\hss}\box0}}
{\setbox0=\hbox{$\scriptscriptstyle\rm O$}\hbox{\hbox to0pt
{\kern0.4\wd0\vrule height0.9\ht0\hss}\box0}}}}
\newcommand{\ca}{\mathcal A}

\newcommand{\cb}{\mathcal B}

\newcommand{\ch}{\mathcal H}

\newcommand{\cn}{\mathcal N}

\newcommand{\calr}{\mathcal R}

\newcommand{\ct}{\mathcal T}

\newcommand{\fp}{\mathfrak{p}}

\newcommand{\intsum}{\mathclap{\displaystyle\int}\mathclap{\textstyle\sum}}

\def\nn{\nonumber}

\newcommand{\eqa}{\begin{eqnarray}}
\newcommand{\neqa}{\end{eqnarray}}

\def\f{\frac}

\usepackage{bbm}

\def\q{{\quad}}

\definecolor{darkgreen}{rgb}{0.0, 0.5, 0.13}

\newcommand{\ketbra}[2] {
	| #1 \rangle \! \langle #2 |}

\begin{document}

\title{Quantum Relativity of Subsystems}

\author{Shadi Ali~Ahmad}
\email[]{shadi.ali.ahmad.22@dartmouth.edu}
\affiliation{Department of Physics and Astronomy, Dartmouth College, Hanover, New Hampshire 03755, USA}

 \author{Thomas D. Galley}
\email{tgalley1@perimeterinstitute.ca}
\affiliation{Perimeter Institute for  Theoretical Physics, 31 Caroline St N, Waterloo, Ontario, N2L 2Y5 Canada}

\author{Philipp A. H\"{o}hn}
\email[]{philipp.hoehn@oist.jp}
\affiliation{Okinawa Institute of Science and Technology Graduate University, Onna, Okinawa 904 0495, Japan}
\affiliation{Department of Physics and Astronomy, University College London, London, WC1E 6BT, United Kingdom}

\author{Maximilian P. E. Lock}
\email[]{maximilian.lock@oeaw.ac.at}
\affiliation{Atominstitut, Technische Universit\"{a}t Wien, 1020 Vienna, Austria}
\affiliation{Institute for Quantum Optics and Quantum Information (IQOQI), Austrian Academy of Sciences, 1090 Vienna, Austria}

\author{Alexander R. H. Smith}
\email[]{arhsmith@anselm.edu}
 \affiliation{Department of Physics, Saint Anselm College, Manchester, New Hampshire 03102, USA} \affiliation{Department of Physics and Astronomy, Dartmouth College, Hanover, New Hampshire 03755, USA}

\date{\today}		

\begin{abstract}
One of the most basic notions in physics is the partitioning of a system into subsystems, and the study of correlations among its parts. In this work, we explore this notion in the context of quantum reference frame (QRF) covariance, in which this partitioning is subject to a symmetry constraint. We demonstrate that different reference frame perspectives induce different sets of subsystem observable algebras, which leads to a gauge-invariant, frame-dependent notion of subsystems and entanglement. 
We further demonstrate that subalgebras which commute \emph{before} imposing the symmetry constraint can translate into non-commuting  algebras in a given QRF perspective \emph{after} symmetry imposition. Such a QRF perspective does  not inherit the  distinction between subsystems in terms of the corresponding tensor factorizability of the kinematical Hilbert space and observable algebra. Since the condition for this to occur is contingent on the choice of QRF, the notion of subsystem locality  is frame-dependent. 
\end{abstract}

\maketitle

\textit{Introduction.---}
Operationally, subsystems are distinguished by physically accessible measurements. Suppose that one can measure the set of observables described by the minimal algebra $\ca$ containing a collection $\{ \mathcal{A}_i\}_{i=1}^n$ of commuting subalgebras, $[ \mathcal{A}_i, \mathcal{A}_j ] =0$ for $i\neq j$.
This implies that observables in $\mathcal{A}_i$ and $\mathcal{A}_j$ are simultaneously measurable, as expected of observables associated with distinct subsystems. When these algebras admit Hilbert space representations $\mathcal{A}\simeq \mathcal{B}(\mathcal{H})$ and $\mathcal{A}_i\simeq \mathcal{B}(\mathcal{H}_i)$, the commuting subalgebra structure can induce a tensor product structure on the composite Hilbert space $\mathcal{H} \simeq \bigotimes_{i=1}^n \mathcal{H}_i$, where $\mathcal{H}_i$ is associated with the $i$-th subsystem. Given that this tensor product structure is induced by the distinguished sets of observables $\mathcal{A}_i$, entanglement and the notion of subsystem itself is defined relative to these distinguished sets~\cite{zanardi2001virtual,barnum_generalizations_2003,barnum_subsystem-independent_2004,zanardiQuantumTensorProduct2004,viola_barnum_2010}.

The physically accessible observables and states of a system are dictated by the symmetries of the situation under consideration \cite{zanardi2001virtual,barnum_generalizations_2003,barnum_subsystem-independent_2004,zanardiQuantumTensorProduct2004,viola_barnum_2010,bartlett_reference_2007}. For example, in a gauge theory physically accessible states and observables are invariant with respect to arbitrary gauge transformations~\cite{Henneaux:1992ig}. In the canonical approach, going back to Dirac~\cite{DiracLectures64}, this invariance requirement is implemented by introducing a kinematical Hilbert space $\mathcal{H}_{\rm kin} \simeq \bigotimes_{i=1}^n \mathcal{H}_i$ that may come equipped with a kinematical tensor product structure. Supposing that $\hat{C} \in \mathcal{L}(\mathcal{H}_{\rm kin})$ is a generator of a gauge symmetry, the physical, \textit{i.e.}\ invariant states satisfy the constraint equation $\hat{C} \ket{\psi_{\rm phys}} = 0$. This is necessary, for example, for the counting of independent gauge-invariant degrees of freedom (such as photon polarizations).
Solutions to this equation may lie outside of $\mathcal{H}_{\rm kin}$ as they may not be normalizable with respect to its inner product. To overcome this issue, one introduces a new physical inner product that is used to complete the solution space of the constraint equation to form the physical Hilbert space $\mathcal{H}_{\rm phys}$ (\textit{e.g.}\ \cite{rovelli_2004,thiemannModernCanonicalQuantum2008,marolfGroupAveragingRefined2002,marolfRefinedAlgebraicQuantization1995,Giulini:1998rk,Giulini:1998kf}). Physical observables are elements of the physical algebra $\mathcal{A}_{\rm phys} \simeq \mathcal{B}(\mathcal{H}_{\rm phys})$ known as Dirac observables, and commute with the constraint on physical states, $[\mathcal{A}_{\rm phys}, \hat{C} ] \ket{\psi_{\rm phys}} = 0$. This ensures that $\mathcal{A}_{\rm phys}$ is invariant under gauge transformations generated by $\hat{C}$, which is necessary for the gauge-invariance of physical expectation values. This constraint based approach also applies to operational scenarios without \emph{bona fide} gauge symmetry, where these constraints correspond to an agent using an internal quantum system as reference frame instead of an external classical one~\cite{Krumm:2020fws}.

It is important to note that the physical Hilbert space $\mathcal{H}_{\rm phys}$ does not inherit the kinematical tensor product structure  $\mathcal{H}_{\rm kin} \simeq \bigotimes_{i=1}^n \mathcal{H}_i$ and associated notion of subsystem. Instead, a notion of subsystem must be induced by commuting subalgebras of $\mathcal{A}_{\rm phys}$, and in general, will be nonlocal with the respect to the kinematical tensor product structure.\footnote{For clarity, ``local'' in this context refers to an operator acting non-trivially only on some tensor factor of a Hilbert space.}

In this letter, we consider composite systems that are invariant under a gauge transformation admitting a tensor product representation across $\mathcal{H}_{\rm kin}$: that is, gauge transformations that act locally on the kinematical factors $\mathcal{H}_i$. We take one of these kinematical subsystems to serve as a reference frame from which the remaining subsystems are described. To do so, we make use of recent results from the theory of quantum reference frames (QRFs) to transform from the quantum theory on the physical Hilbert space $\mathcal{H}_{\rm phys}$, which encodes all QRF choices, to an isomorphic theory from the perspective of a subsystem serving as a reference frame~\cite{vanrietvelde2020change,Vanrietvelde:2018dit,hohn2020switch,Hoehn:2018whn,hoehn2019trinity,HLSrelativistic,periodtrinity,Castro-Ruiz:2019nnl,Krumm:2020fws,Giacomini:2020ahk,Giacomini:2021gei} (see also \cite{Giacomini:2017zju,Giacomini:2019fvi,delaHamette:2020dyi,streiter2020relativistic,Ballesteros:2020lgl} for a related formulation without constraints). We show that subsystems encoded by a perspective-dependent tensor product structure induce a partitioning of the physical Hilbert through the construction of sets of commuting subalgebras of so-called relational Dirac observables \cite{rovelli_2004,thiemannModernCanonicalQuantum2008,dittrich_partial_2006,dittrich_partial_2007,Rovelli:1989jn,Rovelli:1990ph,Rovelli:1990pi,Rovelli:1990jm,hoehn2019trinity,HLSrelativistic,periodtrinity,Krumm:2020fws,vanrietvelde2020change,Vanrietvelde:2018dit,hohn2020switch,Hoehn:2018whn,Chataignier:2020fys,Chataignier:2019kof} associated with different reference frame perspectives. In general, different reference frames induce different partitions of the physical Hilbert space and invariant observable algebra, and thus the associated notion of subsystems is reference frame-dependent. We identify the necessary and sufficient condition for when the physical Hilbert space inherits (some of) the kinematical subsystem partitioning in terms of the spectrum of the relevant constraint, and this condition is contingent on the choice of QRF. This allows us to develop a description of subsystems and entanglement in terms of physical Hilbert space structures that is manifestly gauge invariant and reference frame-dependent.

\textit{From physical states to QRF perspectives.---} 
Consider a kinematical Hilbert space that partitions into three factors ${\mathcal{H}_\mathrm{kin}=\mathcal{H}_\mathrm{A}\otimes\mathcal{H}_\mathrm{B}\otimes\mathcal{H}_\mathrm{C}}$, and a single constraint of the form $\hat{C}=\hat{C}_\mathrm{A}+\hat{C}_\mathrm{B}+\hat{C}_\mathrm{C}$, where each $\hat C_i$ is self-adjoint and acts only on $\mathcal{H}_i$. For simplicity, we assume that each $\hat C_i$ is nondegenerate, and treat degeneracies in the Supplemental Material. Each $\hat C_i$ thus generates a unitary representation of either the translation group $\mathbb{R}$ or $\rm{U}(1)$ on $\ch_i$~\cite{procesi2007lie}.  Consequently, $\hat C$ generates a one-parameter unitary representation of either $\mathbb{R}$ or $\rm{U}(1)$ on $\ch_{\rm kin}$, depending on the combination of the $\hat C_i$ \cite{periodtrinity}. The constraint $\hat C$ may be a Hamiltonian constraint as in gravitational systems, generating temporal reparametrization invariance and dynamics \cite{rovelli_2004,thiemannModernCanonicalQuantum2008,bojobuch,page_evolution_1983,wootters_time_1984, hohn2020switch,Hoehn:2018whn,Smith:2017pwx,Smith:2019imm, hoehn2019trinity,HLSrelativistic,periodtrinity,dittrich_partial_2006,dittrich_partial_2007,Rovelli:1989jn,Rovelli:1990jm,Rovelli:1990ph,Rovelli:1990pi,Chataignier:2020fys,Chataignier:2019kof}, or it may be the generator of a spatial symmetry, such as spatial translation invariance \cite{vanrietvelde2020change,Vanrietvelde:2018dit}. For simplicity, we do not consider interactions between the subsystems $A,B,C$ in the constraint.

As noted above, physical states satisfy $\hat{C} \ket{\psi_{\rm phys}} = 0$, and together they constitute $\mathcal{H}_{\rm phys}$. This induces a redundancy with respect to $\mathcal{H}_\mathrm{kin}$, which can be removed by identifying the choice of redundant subsystem with the choice of QRF relative to which the other systems will be described. A physical state encodes each choice of QRF, therefore assuming the role of a \emph{perspective-neutral} state, linking all the different perspectives~\cite{vanrietvelde2020change,Vanrietvelde:2018dit,hohn2020switch,Hoehn:2018whn,hoehn2019trinity,HLSrelativistic,periodtrinity}. Letting $i,j,k\in\lbrace \mathrm{A,B,C} \rbrace$, we denote the chosen reference system by $k$ and the remaining kinematical factors by $i$ and $j$. We then define 
\begin{align}
\sigma_{ij \vert k} & \ce \mathrm{spec} \left( \hat{C}_{i}+\hat{C}_{j}  \right) \cap \mathrm{spec} \left( -\hat{C}_{k} \right) , \label{eDoubCond}
\end{align}
allowing us to write an arbitrary physical state as
\begin{equation}\label{physk}
    \ket{\psi_{\rm phys}}=\:\: \intsum_{\:\:\: c_i+c_j\in\sigma_{ij\vert k}}\psi(c_i,c_j)\ket{-c_i-c_j}_k\otimes\ket{c_i}_i\otimes\ket{c_j}_j
\end{equation}
for some $\psi(c_i,c_j)$, where $\ket{c_i}_{i}$ is the eigenstate of $\hat{C}_i$ with eigenvalue $c_i$ (likewise for $j$ and $k$). Thus if $\psi(c_i,c_j)$ has non-trivial support over various values of the eigenvalues $c_i,c_j$, then $k$ is entangled with $i,j$ relative to the \emph{kinematical} tensor product structure. However, due to the redundancy, this entanglement is not gauge-invariant~\cite{hoehn2019trinity}.

We can then describe physics from $k$'s perspective via either of two paths~\cite{hoehn2019trinity,HLSrelativistic,periodtrinity}: a ``relational Schr\"{o}dinger picture'' (known in the context of Hamiltonian constraints as the Page-Wootters formalism ~\cite{page_evolution_1983,wootters_time_1984}), and a ``relational Heisenberg picture''~\cite{vanrietvelde2020change,Vanrietvelde:2018dit,hohn2020switch,Hoehn:2018whn}. In both cases, observables on $i,j$ are described relative to outcomes of an observable on $k$, namely an element of a positive operator-valued measure (POVM). The elements of this POVM can be constructed via projectors onto \emph{orientation} states of the reference frame:
\begin{equation} \label{eOrientStates}
    \ket{g}_k\ce  \:\:\:\; \intsum_{\:\:\:c_k}\, e^{i[\theta(c_k)-c_k g]}\ket{c_k}_k,
\end{equation}
where $\theta(c_k)$ are arbitrary phases and $g$ is a coordinate on $G_k$, the group generated by $\hat C_k$. These orientation states transform covariantly under $G_k$, $\ket{g'} = e^{-i(g'-g)\hat{C}_k}\ket{g}$~\cite{buschOperationalQuantumPhysics,holevoProbabilisticStatisticalAspects1982,hoehn2019trinity,HLSrelativistic,periodtrinity,Smith:2017pwx,Smith:2019imm}. The QRF perspective corresponding to $k$ is obtained by conditioning physical states on $k$ being in the orientation $g$, thus fixing the gauge, leading to a \textit{reduced physical Hilbert space} $\mathcal{H}_{ij \vert k}$. In the relational Schr\"{o}dinger picture, this proceeds via the reduction map $\mathcal{R}^{(S)}_k (g): \mathcal{H}_\mathrm{phys} \to \mathcal{H}_{ij \vert k}$ given by $\mathcal{R}^{(S)}_k (g) \ce \bra{g}_{k}\otimes\mathbf{1}_{ij}$ (with its domain restricted to $\mathcal{H}_\mathrm{phys}$). This leads to the orientation-dependent relational Schr\"{o}dinger state $\ket{\psi_{ij|k}(g)} \ce \mathcal{R}^{(S)}_k(g) \ket{\psi_{\rm phys}}\in\ch_{ij\vert k}$ and the decomposition $\ket{\psi_{\rm phys}} = \mu\int_{G_k} dg \, \ket{g}_k \otimes\ket{\psi_{ij|k}(g)}$ exhibiting the kinematical entanglement between $k$ and $i,j$, where $\mu$ is a normalization factor.

On the other hand, in the relational Heisenberg picture one first transforms $\ket{\psi_{\rm phys}}$ to shift the non-redundant information into the $i,j$ partition with a \emph{frame disentangler} (`trivialization') that is a shift conditional on frame $k$, $\ct_{k,\varepsilon}\ce\mu\int_{G_k} dg\,e^{i\varepsilon_k g}\ket{g}\!\bra{g}_k\otimes e^{i(\hat C_i+\hat C_j)g}$. This factors out the QRF, removing the kinematical entanglement between $k$ and $i,j$ (see Supplemental Material):
\begin{equation}
    \ct_{k,\varepsilon}\,\ket{\psi_{\rm phys}}
=\ket{\varepsilon_k}_k \otimes\ket{\psi_{ij\vert k}},\label{dis}
\end{equation}
where $\otimes$ denotes the kinematical tensor product between $k$ and $i,j$,
$\ket{\psi_{ij\vert k}}=e^{i(\hat C_i+\hat C_j)g}\ket{\psi_{ij\vert k}(g)}\in\ch_{ij\vert k}$ is the corresponding ``relational Heisenberg state'', and $\ket{\varepsilon_k}_k=\mu\int_{G_k}dg\,e^{i\varepsilon_kg}\ket{g}_k$. Here, $-\varepsilon_k$ must be a fixed, but arbitrary element of $\sigma_{ij\vert k}$ in which case Eq.~\eqref{dis} satisfies the transformed constraint $\ct_{k,\varepsilon}\,\hat C\,\ct_{k,\varepsilon}^{-1}\sim (\hat C_k-\varepsilon_k\mathbf{1})$, which fixes the now-redundant QRF $k$ and preserves gauge-invariance. One then conditions on the reference frame being in a given orientation of $k$, leading to the relational Heisenberg picture reduction map $\mathcal{R}^{(H)}_{k}: \mathcal{H}_\mathrm{phys} \to \mathcal{H}_{ij \vert k}$ given by
\begin{equation}
    \calr^{(H)}_{k} \ce  \mathcal{R}^{(S)}_k (g) \cn (g,\varepsilon_k)  \ct_{k,\varepsilon},\label{qcm}
\end{equation}
 where $\cn (g,\varepsilon_k)$ is a normalization factor such that $\calr^{(H)}_{k}\ket{\psi_{\rm phys}}=\ket{\psi_{ij\vert k}}$, which is (weakly) independent of $g$ and $\varepsilon$, and we therefore do not include these labels in $\calr^{(H)}_{k}$. This reduction is unitarily equivalent to acting with $\mathcal{R}^{(S)}_k(g)$ on \emph{physical states}~\cite{hoehn2019trinity,HLSrelativistic,periodtrinity}. We will use the relational Heisenberg picture in what follows, denoting the reduction map by $\mathcal{R}_k \equiv\mathcal{R}_k^{(H)}$ for simplicity.

The reduced physical Hilbert space $\ch_{ij\vert k}$ and observables on it encode the physics of $i,j$ as described from the internal perspective of QRF $k$. When $G_k=\rm{U}(1)$, $\ch_{ij\vert k}$ need not be a subspace of the kinematical factors $\ch_i\otimes\ch_j$~\cite{periodtrinity}. Furthermore, thanks to the redundancy in describing $\ch_{\rm phys}$, the reduction is invertible on physical states (but not on $\mathcal{H}_\mathrm{kin}$), so that $\ch_{ij\vert k}$ is isometric to $\ch_{\rm phys}$ \cite{vanrietvelde2020change,Vanrietvelde:2018dit,hohn2020switch,Hoehn:2018whn,hoehn2019trinity,HLSrelativistic,periodtrinity}. Hence, the algebraic properties of observables are preserved. This permits us to change QRF: the change from $k$ to $i$ takes the compositional form of a ``quantum coordinate transformation'', $\Lambda_{k\to i}\ce\calr_{i}\circ\calr^{-1}_{k}$, transforming both states and observables via the structure on $\ch_{\rm phys}$ which is \emph{a priori} neutral with respect to QRF perspectives~\cite{vanrietvelde2020change,Vanrietvelde:2018dit,hohn2020switch,Hoehn:2018whn,hoehn2019trinity,HLSrelativistic,periodtrinity}, see Fig.~\ref{ChangingRF}. The same physical situation, encoded in the perspective-neutral state $\ket{\psi_{\rm phys}}$, is thus described from different internal QRF perspectives. We shall now exploit this gauge-invariant, perspective-neutral framework to explain dependence of, first,  subsystem locality and correlations and, second, tensor factorizability on the choice of QRF.

\begin{figure}[t]
\includegraphics[width= 245pt]{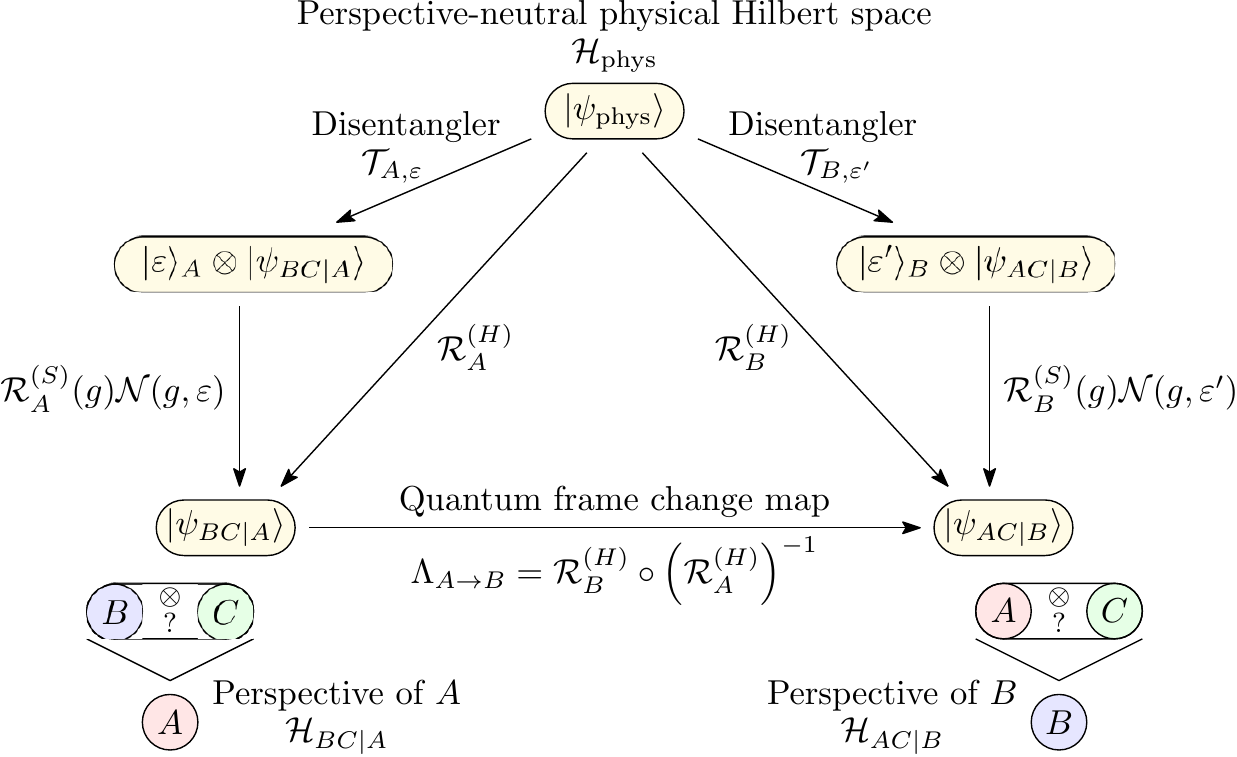}
\caption{The change from perspective $A$ to perspective $B$ takes the compositional form of a `quantum coordinate transformation', $\Lambda_{A\to B}\ce\calr_{B}^{(H)}\circ (\calr^{(H)}_{A})^{-1}$~ \cite{vanrietvelde2020change,Vanrietvelde:2018dit,hohn2020switch,Hoehn:2018whn,hoehn2019trinity,HLSrelativistic,periodtrinity}. This induces a transformation on the algebra observables from the perspective of $A$ to the perspective of $B$, namely $\Lambda_{A\to B} \mathcal{A}_{BC|A} \Lambda_{A\to B}^{-1} \subseteq \mathcal{A}_{AC|B} $, where $\mathcal{A}_{ij|k} \simeq \mathcal{B}(\mathcal{H}_{ij|k})$.  \label{ChangingRF}}
\end{figure}

\textit{Frame-dependent subsystems and correlations.}--- 
The QRF dependence of correlations has been observed in~\cite{Giacomini:2017zju}, with the conclusion that superposition in one frame manifests as entanglement in another frame. Later, the formalism for changing QRFs introduced in~\cite{Giacomini:2017zju} was shown to be equivalent to the frame-change map in Fig.~\ref{ChangingRF}~\cite{vanrietvelde2020change}, and the frame dependence of correlations was studied in a variety of contexts~\cite{Krumm:2020fws,vanrietvelde2020change,Castro-Ruiz:2019nnl,hoehn2019trinity,delaHamette:2020dyi}. Here, we use the perspective-neutral architecture to describe the QRF relativity of subsystems and correlations.

The Heisenberg-picture reduction illustrates why a fixed perspective-neutral state $\ket{\psi_{\rm phys}}$ generally leads to different properties, such as correlations, in $A$'s and $B$'s perspective (see Fig.~\ref{ChangingRF}): when going to $A$'s perspective, the now-redundant $A$ becomes kinematically disentangled from the non-redundant $B,C$, while $B$ becomes disentangled from $A,C$ when proceeding to $B$'s perspective. Which kinematical tensor factor in $\ket{\psi_{\rm phys}}$ is chosen as redundant and which as independent changes. In other words, the non-redundant (physical) information in $\ket{\psi_{\rm phys}}$ is shifted among different kinematical tensor factors when going to different QRF perspectives. Indeed, the wave function $\psi(c_i,c_j)$ will look different for different choices of $i,j$. As we shall see later, $\ch_{ij\vert k}$ may not even be factorizable across $i$ and $j$.

However, when there \emph{is} a physical tensor product structure this generically leads to different correlations in different frames. This can be understood by examining the observables that probe the respective tensor factorizations. Suppose the reduced physical Hilbert space $\ch_{BC|A} \simeq \ch_{B|A} \otimes \ch_{C|A}$ admits a tensor factorization across the subsystems $B,C$ from $A$'s perspective, induced from the original tensor product structure of $\ch_{\rm kin}$, and similarly that $\ch_{AC|B} \simeq \ch_{A|B} \otimes \ch_{C|B}$ so that we can consider entanglement across these subsystems. 
The subsystem physical Hilbert spaces $\ch_{i|k}$ may be different from their kinematical counter-parts $\ch_i$ of which they may \cite{hoehn2019trinity,HLSrelativistic} or may not be \cite{periodtrinity} subspaces.
We consider the algebra generated by local subsystem observables on $\ch_{ij|k}$, namely $\ca_{ij|k} \ce \ca_{i\vert k}\otimes\ca_{j\vert k}$, where $\ca_{i\vert k}\ce\cb(\ch_{i\vert k})$ (hence a type I factor), so that $[\ca_{i\vert k}\otimes\mathbf{1}_{j\vert k},\mathbf{1}_{i\vert k}\otimes\ca_{j\vert k}]=0$. 
Since $\ca_{ij\vert k}$ is dense in $\cb(\ch_{ij\vert k})$ with respect to the strong operator topology
\cite{takesaki1,parthasarathy2012introduction}, we can treat $\ca_{ij\vert k}$ for all practical purposes as the observable algebra of the tensor product space $\ch_{ij\vert k}$ (in finite dimensions the algebras are isomorphic).

Using the fact that $\calr_k$ is an invertible isometry, these observable algebras can be embedded into the algebra of relational Dirac observables $\ca_{\rm phys}\ce \cb(\ch_{\rm phys})$ as $\ca_{\rm phys}^{i\vert k}\ce\calr^{-1}_{k}\left(\ca_{i\vert k}\otimes\mathbf{1}_j\right)\calr_{k}\subset\ca_{\rm phys}$. This yields  $\big[\ca_{\rm phys}^{i\vert k},\ca_{\rm phys}^{j\vert k}\big]=0$.  Since the properties of $\calr_{k}$ thus imply that $\ca_{\rm phys}^{i\vert k}$ and $\ca_{\rm phys}^{j\vert k}$ are commuting type I factors, and that moreover $\calr^{-1}_{k} \,\ca_{ij\vert k}\,\calr_{k}=\ca_{\rm phys}^{i\vert k}\cdot\ca_{\rm phys}^{j\vert k}$ is dense in $\ca_{\rm phys}$, it follows  \cite{takesaki1}
that these factors induce a \emph{physical} tensor product on $\mathcal{H}_{\rm phys}\simeq\ch_{\rm phys}^{i\vert k}\otimes\ch_{\rm phys}^{j\vert k}$.  The following theorem shows this to be QRF-dependent in general (see Supplemental Material, also for the degenerate case). 
\begin{Theorem}\label{lem_1}
The algebra  $\mathcal{A}^{C|A}_{\rm phys}$ of relational observables of $C$ relative to $A$, is distinct from the algebra  $\mathcal{A}^{C|B}_{\rm phys}$  of relational observables of $C$ relative to $B$,  $\ca_{\rm phys}^{C\vert A}\neq\ca_{\rm phys}^{C\vert B}$.
\end{Theorem}
However, note that $\ca^{C\vert A}_{\rm phys}$ is isomorphic to $\ca^{C\vert B}_{\rm phys}$ when $\ch^{C\vert A}_{\rm phys}\simeq\ch^{C\vert B}_{\rm phys}$. The two tensor factorizations $\ch_{\rm phys}\simeq \ch_{\rm phys}^{B\vert A}\otimes\ch_{\rm phys}^{C\vert A}$ and $\ch_{\rm phys}\simeq \ch_{\rm phys}^{A\vert B}\otimes\ch_{\rm phys}^{C\vert B}$ therefore constitute \emph{different} physical tensor factorizations. It is therefore clear that a given physical state $\ket{\psi_{\rm phys}} \in \mathcal{H}_{\rm phys}$ exhibits different correlations in the two different factorizations.

Put differently, transforming the algebra $\ca_{C\vert A}$ from $A$'s to $B$'s perspective, does \emph{not} yield $C$'s algebra relative to $B$, $\Lambda_{A\to B} \left(\mathbf{1}_B\otimes\ca_{C\vert A}\right)\Lambda_{A\to B}^{-1}\neq\mathbf{1}_A\otimes\ca_{C\vert B}$. The tensor factorization between $B,C$ relative to $A$ thus does not map under QRF transformations into the tensor factorization between $A,C$ relative to $B$. Instead, since $\Lambda_{A\to B}$ is an invertible isometry too, it maps into a different tensor factorization relative to $B$, namely one between combinations of $A$ and $C$ degrees of freedom. Consequently, \emph{the notion of subsystem locality is QRF dependent}, as are the correlations inherited from a given physical state. We illustrate this observation in an example in the Supplemental Material.

\textit{When are kinematical subsystems physical?---}  
As noted above, it is in general not the case that the reduced physical Hilbert space and the observable algebra can be factorized into the same subsystems that constitute tensor factors of the  kinematical Hilbert space. In other words, imposing a given symmetry can remove the distinction between what might have been expected to be physical subsystems.

To see this, note that one can also obtain the reduced physical Hilbert space from $k$'s perspective \emph{directly} from the kinematical $i,j$ tensor factors via the (possibly improper) projector $\Pi_{\sigma_{ij \vert k}}:\mathcal{H}_{i} \otimes \mathcal{H}_{j}\to\ch_{ij|k}$ given by~\cite{hoehn2019trinity,HLSrelativistic,periodtrinity}
\begin{equation} \label{ePhysicalProjectorBipartite}
\begin{aligned}
    \Pi_{\sigma_{ij \vert k}} \ce \:\: \intsum_{\:\:\: c_i ,c_j \vert c_i + c_j \in \sigma_{ij \vert k}} \ket{c_{i}} \bra{ c_{i}}_{i} \otimes \ket{c_{j}} \bra{c_{j}}_{j} .
\end{aligned}
\end{equation}
Observe that $\sigma_{ij\vert k}$ is symmetric in $i$ and $j$ but not in $i$ and $k$. Furthermore, defining $\sigma_i \ce  \mathrm{spec} ( \hat{C}_{i} )$, note that $\sigma_{ij \vert k} = \spec (\hat{C}_{i} + \hat{C}_{j})$ whenever $\sigma_{k}=\mathbb{R}$. The projector $\Pi_{\sigma_{ij\vert k}}$ is improper if $\sigma_{ij\vert k}$ is discrete, while at least one of $\hat C_i,\hat C_j$ has continuous spectrum \cite{periodtrinity}. The reduced physical Hilbert space $\ch_{ij|k}$ factorizes into $i$ and $j$ subsystems $\ch_{i|k}$ and $\ch_{j|k}$ if and only if $\Pi_{\sigma_{ij\vert k}}$ does as well. This is only the case  if (see Supplemental Material for proof)
\begin{equation} \label{eCondition}
\sigma_{ij \vert k} = \mathrm{M} \left( \sigma_{i\vert jk} , \sigma_{j\vert ik} \right) ,
\end{equation}
where $\sigma_{i\vert jk} \ce \sigma_{i} \cap \spec \left(-\hat{C}_{j}-\hat{C}_{k} \right)$ is the subset of $\sigma_{i}$ compatible with the constraint equation, and where $\mathrm{M}(\cdot,\cdot)$ denotes Minkowski addition, defined by $\mathrm{M}(X,Y) \ce \lbrace x+y \vert x\in X, y\in Y \rbrace$. When this is satisfied, then the projector defined in Eq.~\eqref{ePhysicalProjectorBipartite} becomes
\begin{equation} 
\Pi_{\sigma_{ij \vert k}} = \left( ~~ \intsum_{\:\:\: c_{i} \in \sigma_{i\vert jk}} \ket{c_{i}}\bra{c_{i}}_i \right)  \otimes \left ( ~~ \intsum_{\:\:\: c_{j} \in \sigma_{j\vert ik}} \ket{c_{j}}\bra{c_{j}}_j \right),  
\end{equation}
and thus $\Pi_{\sigma_{ij \vert k}} \left( \mathcal{H}_i \otimes \mathcal{H}_j \right)  = \mathcal{H}_{i\vert k} \otimes \mathcal{H}_{j\vert k}$. Here, $\mathcal{H}_{i\vert k}  \subseteq \mathcal{H}_i$, unless $\hat C_k$ has discrete and $\hat C_i$ continuous spectrum \cite{hoehn2019trinity,HLSrelativistic,periodtrinity} (likewise for $\ch_{j\vert k}$).  Note that when $\sigma_{i\vert jk} = \sigma_i$, then $\ch_{i\vert k} =\mathcal{H}_i$. This holds for both $i$ and $j$ if $\sigma_k =\mathds{R}$. We give an example of non-factorizability of the physical Hilbert space in the Supplemental Material.

To understand this more explicitly, let $\hat{A}_{i} \otimes \hat{A}_{j}$ be a kinematical basis element of $\mathcal{B}(\mathcal{H}_{i}) \otimes \mathcal{B}(\mathcal{H}_{j})$. Condition~\eqref{eCondition} must be satisfied in order for the physical representation of this operator from $k$'s perspective to factorize across $i$ and $j$. Otherwise, the degrees of freedom of $i$ and $j$ become combined indivisibly into $\Pi_{\sigma_{ij \vert k}} ( \hat{A}_{i}\otimes\hat{A}_{j} )$. This includes the case when $\hat{A}_i \otimes \hat{A}_j$ is diagonal in the eigenbases of $\hat C_i$ and $\hat C_j$, and thus commutes with $\Pi_{\sigma_{ij\vert k}}$ (see Supplemental Material). Specifically, kinematical $i$ subsystem observables of the form $\hat A_i\otimes\mathbf{1}_j$ will \emph{not} translate into a product form on $\ch_{ij\vert k}$.  In fact (see Supplemental Material): 
\begin{Theorem}\label{thm_2}
There exist $\hat A_i\otimes\mathbf{1}_j$ and $\mathbf{1}_i\otimes\hat A_j$ in $\cb(\ch_i)\otimes\cb(\ch_j)$ whose images under $\Pi_{\sigma_{ij\vert k}}$ in $\cb(\ch_{ij\vert k})$ do \emph{not} commute unless condition~\eqref{eCondition} is met.
\end{Theorem}
By linearity, these conclusions extend to an arbitrary element of $\mathcal{B}(\mathcal{H}_{i}) \otimes \mathcal{B}(\mathcal{H}_{j})$. Consequently, when condition~\eqref{eCondition} is not satisfied, the algebra of observables loses its distinction between  parties $i$ and $j$  from the perspective of $k$'s reference frame. 

\textit{Frame-dependent factorizability.---} 
When Eq.~\eqref{eCondition} holds in one frame but not in another, the preservation of the kinematical factorization on $\mathcal{H}_{ij\vert k}$ likewise depends on the frame, as we now illustrate. For concreteness, consider any constraint such that $\sigma_\mathrm{A}=\mathds{R}_{+}$, $\sigma_\mathrm{B} =\mathds{R}_{+}$, and $\sigma_\mathrm{C} =\mathds{R}$. Considering first $C$'s perspective, we have that $\sigma_{\mathrm{AB} \vert \mathrm{C}}=\mathds{R}_{+}=\mathrm{M}(\mathds{R}_{+},\mathds{R}_{+})$, \textit{i.e.}\ condition~\eqref{eCondition} is satisfied, with $\sigma_{\mathrm{A}\vert\mathrm{BC}}=\sigma_\mathrm{A}$ (likewise $\sigma_{\mathrm{B}\vert\mathrm{AC}}$), and $\Pi_{\sigma_{\mathrm{AB} \vert \mathrm{C}}} ( \mathcal{H}_\mathrm{A} \otimes \mathcal{H}_\mathrm{B} ) = \mathcal{H}_\mathrm{A} \otimes \mathcal{H}_\mathrm{B}$. From $B$'s perspective, on the other hand, one can prove by contradiction that condition~\eqref{eCondition} is not satisfied (see Supplemental Material), and therefore the reduced physical Hilbert space does not factor into $A$ and $C$ parts. This latter fact, does not however imply that there exists no tensor factorization of $\ch_{AC|B}$. Indeed, one can use the tensor factorization of $\ch_{AB|C}$ to construct one on $\ch_{AC|B}$ via the frame-change map $\Lambda_{C\to B}$, as in our discussion of frame-dependent correlations above. In this case the algebra of local observables from $C$'s perspective, namely $\ca_{A\vert C}\otimes\ca_{B\vert C}$, maps to a tensor factorization between \emph{combinations} of $A$ and $C$ degrees of freedom, and therefore does not correspond to a partitioning into subsystems $A$ and $C$. 

As a particular example, consider a reparameterization-invariant system consisting of two free (non-relativistic) unit-mass particles, $A$ and $B$, and an ideal clock $C$ (\textit{i.e.}\ one whose Hamiltonian is equivalent to a momentum operator~\cite{pauli1958allgemeinen,srinivas1981time,woods2019autonomous,hohn2020switch,Smith:2017pwx,Castro-Ruiz:2019nnl}). This corresponds to the constraint 
\begin{equation} \label{eTwoPartsOnePerf}
\hat{C} = \frac{\hat{p}_{A}^{2}}{2} + \frac{\hat{p}_{B}^{2}}{2} + \hat{p}_{C} .
\end{equation}
Thus $\sigma_\mathrm{A}=\mathds{R}_{+}$, $\sigma_\mathrm{B} =\mathds{R}_{+}$, as above, and therefore the distinction between kinematical subsystems survives on $\ch_{AB|C}$, but not on $\ch_{AC|B}^{d_B}$, where $d_B=\pm1$ labels the degeneracy of $\hat p_B^2$ (see Supplemental Material for a discussion of degeneracies). In the Supplemental Material, we illustrate this by examining how the kinematical canonical pairs $(\hat{x}_{i},\hat{p}_{i})$ on $\mathcal{H}_{i}$ and $(\hat{x}_{j},\hat{p}_{j})$ on $\mathcal{H}_{j}$ appear from $B$'s and $C$'s perspectives. We show that their commutation relations are preserved on $\ch_{AB|C}$, but not on $\ch_{AC|B}^{d_B}$, where they yield mutually non-commuting canonical pairs, in line with Theorem~\ref{thm_2}, explaining the absence of a tensor factorization across $A$ and $C$ relative to $B$. The same conclusion holds for the corresponding relational observables on $\ch_{\rm phys}$. This further highlights the distinction between local observables on $\mathcal{H}_\mathrm{kin}$ and the observables in a given physical reference frame. 

Another example of the above class of constraints is obtained by replacing systems $A$ and $B$ in Eq.~\eqref{eTwoPartsOnePerf} with the nondegenerate Hamiltonian $\hat H=\hat{p}^2/2m+a_1\,e^{a_2\,\hat{q}}$ with $a_1 , a_2 >0$. This observation can also be easily extended to the case when $A,B$ are harmonic oscillators.

\textit{Discussion and conclusions.\,---\,} 
We have established a gauge invariant and quantum frame dependent notion of subsystems, locality and correlations using relational observables. We have exploited the perspective-neutral approach to QRF covariance \cite{vanrietvelde2020change,Vanrietvelde:2018dit,hohn2020switch,Hoehn:2018whn,hoehn2019trinity,HLSrelativistic,periodtrinity,Krumm:2020fws,Giacomini:2020ahk,Giacomini:2021gei,Castro-Ruiz:2019nnl}, showing algebraically how different QRF choices necessarily induce distinct tensor factorizations of the physical Hilbert space when the latter admits such structures. Further, we have identified the necessary and sufficient condition for QRF perspectives to inherit the kinematical partitioning of subsystems. Specifically, we have illustrated that the kinematical definition of subsystems may survive in some QRF perspectives, but dissolve in others.

The ensuing QRF dependence of subsystems and entanglement is a particular realization of the proposal for an observer-dependent notion of generalized entanglement put forward in \cite{barnum_generalizations_2003,barnum_subsystem-independent_2004,viola_barnum_2010} in terms of relational observables and QRFs. It is also related to the generic feature of the dependence of entanglement on classical coordinate choices~\cite{PhysRevLett.98.080406}. However, here, as in previous work on QRFs \cite{Giacomini:2017zju,vanrietvelde2020change,hoehn2019trinity,Castro-Ruiz:2019nnl,delaHamette:2020dyi,Krumm:2020fws}, specific choices of coordinates are associated to the internal perspectives of different systems to form quantum reference systems and these may be in superpositions of ``orientations''.  This provides a physical interpretation to the coordinate choices, and in turn the \emph{quantum} relativity of subsystems and entanglement.

Note that the notion of QRF here is physically distinct from one sometimes used in the context of quantum information theory~\cite{bartlett_reference_2007,Palmer:2013zza,smith_quantum_2016,loveridge_relativity_2017,loveridge_symmetry_2018,smith_communicating_2019} (see, e.g.\ the discussion in~\cite{Krumm:2020fws}), where entanglement must be operationally defined relative to observables that are independent of the choice of an external frame not shared by two parties, for example by appending an ancilla system. This approach also resonates with the proposal in \cite{barnum_generalizations_2003,barnum_subsystem-independent_2004,viola_barnum_2010}, but does not involve the adoption of an internal perspective relative to a subsystem through the reduction maps employed here. In particular, the aim in that context is to define an  \emph{external}-frame-independent notion of entanglement, in contrast with our investigation of an \emph{internal}-frame-dependent entanglement.

In the relativistic case, the entanglement between spin and momentum degrees of freedom for relativistic particles \cite{Peres:2002ip,Peres:2002wx}, as well as the  momentum mode decomposition in quantum field theory \cite{birrell1984quantum}, leading to the Unruh effect \cite{Unruh:1976db}, Hawking radiation \cite{Hawking:1974rv}, and particle creation due to the expansion of the Universe \cite{Parker:1968mv}, are also dependent on the choice of spacetime frame. In contrast to QRFs, these coordinate frames are not associated to dynamically evolving quantum systems, but are idealised non-interacting classical entities external to the physics being considered. 

The frame dependence of factorizability  demonstrated in this letter implies that frameworks for general physical theories which take system composition as a primitive concept~\cite{2001quant.ph..1012H,PhysRevA.75.032304,Chiribella_2010,Gogioso_2018,Mueller:2020yrl,Hoehn:2014uua} are not currently able to describe fully general physical scenarios with multiple frames. 

Finally, it will be fruitful to connect our observations with the currently widely explored notions of local subsystems and entanglement in gauge theories and gravity  \cite{Donnelly:2016auv,Donnelly:2016rvo,Donnelly:2017jcd,Giddings:2018koz,Freidel:2020xyx,Geiller:2019bti,Gomes:2018shn,Gomes:2019xto,Wieland:2017cmf,Wieland:2017zkf,Jacobson:2019gnm,Wong:2017pdm}. For example, defining local subsystems in gravity non-perturbatively in terms of commuting subalgebras of relational observables can complement the perturbative investigation of subsystems in terms of dressed observables in \cite{Donnelly:2016rvo,Donnelly:2017jcd,Giddings:2018koz}. Conversely, the possible non-factorizability of the physical Hilbert space observed here calls for a revision of the notion of subsystems. This question seems to be related to the construction of entangling products using edge modes and extended Hilbert spaces in gauge theories \cite{Geiller:2019bti,Wong:2017pdm,Donnelly:2011hn,Buividovich:2008gq,Casini:2013rba,Delcamp:2016eya}.

\textit{Acknowledgments.\,---\,} PAH thanks Isha Kotecha, Fabio Mele and Markus M\"uller for helpful discussions and is grateful for support from the
Foundational Questions Institute through Grant number
FQXi-RFP-1801A. MPEL acknowledges financial support by the ESQ (Erwin Schrödinger Center for Quantum Science \& Technology) Discovery programme, hosted by the Austrian Academy of Sciences (ÖAW). This work was supported in part by funding from  Okinawa Institute of Science and Technology Graduate University. This research was supported by Perimeter Institute for Theoretical Physics. Research at Perimeter Institute is supported by the Government of Canada through the Department of Innovation, Science and Economic Development Canada and by the Province of Ontario through the Ministry of Research, Innovation and Science. This work was supported in part by a Kaminsky Undergraduate Research Award. ARHS wishes to acknowledge many productive and enjoyable conversations with Lorenza Viola on applications of generalized entanglement in various relational contexts. ARHS is grateful for support from Saint Anselm College and Dartmouth College.

 \textit{Contributions.\,---\,} P. A. H. proved theorems 1 and 2. M. P. E. L. established the condition relating kinematical and physical factorization. All authors contributed actively through discussions and revisions of technical aspects, as well as in the preparation of the manuscript.



\bibliography{RelFact}


\onecolumngrid

\pagebreak

\section{Supplemental Material}

\subsection{From physical states to Quantum Reference Frames (QRFs), including degeneracy}
 
Here we describe how the formalism introduced in the main text generalizes to the case where any of the terms in the constraint $\hat{C}=\hat{C}_\mathrm{A}+\hat{C}_\mathrm{B}+\hat{C}_\mathrm{C}$ may have an (eigenvalue-independent) degeneracy. For example, in relativistic cases, the degeneracy sectors may correspond to positive and negative frequency modes \cite{hohn2020switch,Hoehn:2018whn,HLSrelativistic,Smith:2019imm,Chataignier:2020fys}. In the following sections, all proofs of statements in the main text will be given in a form including the possibility of degeneracies. 

First note that solving $\hat C\,\ket{\psi_{\rm phys}}=0$ for reference system $k$ yields physical states in the form $\ket{\psi_{\rm phys}}=\sum_{d_k}\ket{\psi_{\rm phys}^{d_k}}$ \cite{hoehn2019trinity,HLSrelativistic,periodtrinity}, with
\ba \label{ePhysicalStatesDeg}
\ket{\psi_{\rm phys}^{d_k}}=\sum_{d_i,d_j}\:\: \intsum_{\:\:\: c_i+c_j\in\sigma_{ij\vert k}}\psi^{d_k}_{d_i,d_j}(c_i,c_j) \ket{-c_i-c_j,d_k}_k\otimes\ket{c_i,d_i}_i\otimes\ket{c_j,d_j}_j ,
\ea
for some $\psi^{d_k}_{d_i,d_j}(c_i,c_j)$, where $d_i$ is a degeneracy label for $\hat C_i$. The physical Hilbert space thus decomposes into superselection sectors $\ch_{\rm phys}=\bigoplus_{d_k}\ch_{\rm phys}^{d_k}$~\cite{marolfRefinedAlgebraicQuantization1995,Giulini:1998rk,HLSrelativistic}. The orientation states of the reference frame, introduced in Eq.~\eqref{eOrientStates}, transform covariantly under the group $G_k$ generated by $\hat{C}_{k}$, \textit{i.e.}\ $\exp(-i g'\hat C_k)\ket{g,d_k}_k=\ket{g+g',d_k}_k$. They can therefore be rewritten here as~\cite{hoehn2019trinity,HLSrelativistic,periodtrinity,holevoProbabilisticStatisticalAspects1982,buschOperationalQuantumPhysics,braunsteinGeneralizedUncertaintyRelations1996,Smith:2017pwx,Smith:2019imm}
\ba \label{eOrientStatesDeg}
\ket{g,d_k}_k\ce \:\: \intsum_{\:\:\:c_k}e^{i(\theta(c_k)-c_k g)}\ket{c_k,d_k}_k,
\ea
where $\theta(c_k)$ is an arbitrary phase function. Taken together, the set of all orientation states define a positive operator-valued measure (POVM) over all degeneracy sectors 
\begin{equation}
    E_k(X)\ce\int_{X\subset G_k}E_k(dg) \qquad \text{with} \qquad E_k(dg)=\mu\sum_{d_k}dg\ket{g,d_k}\!\bra{g,d_k},
\end{equation}
where $X$ is a Borel subset of $G_k$ and $\mu$ is a normalization constant such that $\int_G \, E_k(dg) = \mathbf{1}$; if $G_k=\mathbb{R}$ or $G_k=\mathrm{U}(1)$, we have $\mu=1/2\pi$ and $\mu=1/t_\mathrm{max}$, respectively, where $t_\mathrm{max}$ is the period of the $\rm{U}(1)$ representation \cite{hoehn2019trinity,HLSrelativistic,periodtrinity}. This POVM is said to be covariant with respect to $G_k$. 

This orientation POVM can be employed to construct the $n^{\rm th}$-moment operators for the QRF frame orientation
\ba
\hat O_k^{(n)}\ce\int_{G_k}E_k(dg)\, g^n,
\ea
which give rise to a generalization of canonical conjugacy 
\ba
\big[\hat O_k^{(n)},\hat C_k\big]=in\,\hat O^{(n-1)}_k-\frac{is_{G_k}}{\mu^{(n-1)}}\sum_{d_k}\ket{e,d_k}\!\bra{e,d_k},\label{canconj}
\ea
where $s_{\mathbb{R}}=0$ and $s_{\rm{U}(1)}=1$ and $e$ labels the origin/identity of $G_k$ \cite{hoehn2019trinity,HLSrelativistic,periodtrinity}. The outcome of the POVM then provides a parameterization of the group $G$ generated by $\hat C$. It can be used to gauge-fix the $G$-orbits and as a reference degree of freedom for $i,j$. In the case where $G=\mathbb{R}$ and $G_k=\rm{U}(1)$, the QRF orientations cannot parametrize the full orbit of $G$, but there the constraint forces all physical degrees of freedom of $i,j$ to be periodic so that this is not an issue~\cite{periodtrinity}.

To move from $\mathcal{H}_\mathrm{phys}$ to the perspective of $k$, we condition upon the degeneracy label as well as the value $g$ of the orientation observable, which are commuting observables due to the eigenvalue-independence of the degeneracy (see \textit{e.g.}~\cite{HLSrelativistic}). The resulting reduction maps then lead to reduced physical Hilbert spaces $\ch_{ij\vert k}^{d_k}$ and observables on it, encoding the physics of $i,j$ as described from the internal perspective of QRF $k$ within the $d_k$ superselection sector. In particular, the relational Schr\"{o}dinger picture reduction map $\calr_{d_k}^{(S)}(g):\ch_{\rm phys}\rightarrow\ch_{ij\vert k}^{d_k}$ is defined by $\mathcal{R}^{(S)}_{d_k} (g) \ce \bra{g,d_k}_{k}\otimes\mathbf{1}_{ij}$ (acting only on elements of $\ch_\mathrm{phys}$), and writing $\ket{\psi_{ij|k}^{d_k}(g)} \ce \mathcal{R}^{(S)}_{d_k}(g) \ket{\psi_{\rm phys}}$, we have
\begin{equation}
    \ket{\psi_{\rm phys}} = \mu \sum_{d_k} \int_{G_k} dg \, \ket{g,d_k}_k \otimes\ket{\psi_{ij|k}(g)}
\end{equation}

The relational Heisenberg picture reduction map $\calr_{d_k}:\ch_{\rm phys}\rightarrow\ch_{ij\vert k}^{d_k}$ is obtained by first acting with the frame disentangler (or trivialization) of the QRF $k$, defined by \cite{vanrietvelde2020change,Vanrietvelde:2018dit,hohn2020switch,Hoehn:2018whn,hoehn2019trinity,HLSrelativistic,periodtrinity} 
\ba \label{eDisentDeg}
\ct_{k,\varepsilon} &\ce& \sum_{n=0}^\infty \,\f{i^n}{n!}\hat O_k^{(n)}\otimes(\hat C_i\otimes\mathbf{1}_j+\mathbf{1}_i\otimes\hat C_j+\varepsilon_k\mathbf{1}_{ij})^n = \mu \sum_{d_k} \int_{G_k} dg\,e^{i\varepsilon_k g}\ket{g,d_{k}}\!\bra{g,d_{k}}_k\otimes e^{i(\hat C_i+\hat C_j)g}  \\
&\equiv& \sum_{d_k}\ct_{k,\varepsilon}^{d_k} , \nonumber
\ea
where each $\ct_{k,\varepsilon}^{d_k}$ is a $d_k$-sector-wise disentangler, \textit{i.e.}
\begin{equation}
    \ct_{k,\varepsilon}^{d_k}\,\ket{\psi_{\rm phys}}
=\ket{\varepsilon_k,d_{k}}_k \otimes\ket{\psi_{ij\vert k}^{d_k}},
\end{equation}
with $\ket{\psi_{ij\vert k}^{d_k}} \in \ch_{ij\vert k}^{d_k}$ (\textit{cf.} Eq.~\eqref{dis} in the main  text), and where $-\varepsilon_k$ is a fixed, but arbitrary element of $\sigma_{ij\vert k}$. Note that $\varepsilon_k$ must be such that $\ket{\varepsilon_k,d_k}\neq\ket{\varepsilon_k,d_k'} \, \forall \, d_k \neq d_k'$, so that superselection sectors are still distinguishable upon the action of the disentangler, as required for the  latter to be invertible on physical states. As a second step,  one conditions on the QRF $k$ being in orientation $g$, as in the relational Schr\"odinger picture reduction, so that altogether 
\ba 
\calr_{d_k}\ce\,\calr_{d_k}^{(S)}(g) \cn(g,\varepsilon_k)\,\ct_{k,\varepsilon}.
\ea
where $\cn (g,\varepsilon_k)$ is a normalization factor so that $\calr_{d_k}\ket{\psi_{\rm phys}}=\ket{\psi_{ij\vert k}^{d_k}}$ is a relational Heisenberg picture state which neither depends on $g$ nor $\varepsilon_k$. Note that this is unitarily equivalent to the relational Schr\"{o}dinger picture state above, \textit{i.e.}~$\mathcal{R}^{(S)}_{d_k}(g) \ket{\psi_{\rm phys}}=e^{-i(\hat{C}_{i}+\hat C_j)g}\mathcal{R}_{d_k}(g) \ket{\psi_{\rm phys}}$~\cite{hoehn2019trinity,HLSrelativistic,periodtrinity}.

The disentangler $\ct_{k,\varepsilon}$ transforms the constraint $\hat C$ in such a way that it only acts on the QRF $k$ and fixes its now-redundant degrees of freedom $\left(\ct_{k,\varepsilon}\, \hat C\,\ct_{k,\varepsilon}^{-1}\right)\ct_{k,\varepsilon}\ket{\psi_{\rm phys}}=\left(\left(\hat C_k-\varepsilon_k\right)\otimes\mathbf{1}_{ij}\right)\ct_{k,\varepsilon}\ket{\psi_{\rm phys}}$, $\forall\,\ket{\psi_{\rm phys}}\in\ch_{\rm phys}$, when acting on `trivialized' physical states \cite{hoehn2019trinity,HLSrelativistic,periodtrinity}. Here,
\ba 
\ct_{k,\varepsilon}^{-1}\ce\sum_{n=0}^\infty \,\f{(-i)^n}{n!}\hat O_k^{(n)}\otimes(\hat C_i\otimes\mathbf{1}_j+\mathbf{1}_i\otimes\hat C_j+\varepsilon_k\mathbf{1}_{ij})^n,\nn
\ea
which in general is only the inverse of the disentangler when acting on \emph{physical states} \cite{hoehn2019trinity,HLSrelativistic,periodtrinity}, but this is all we need.

The reduction is invertible on each $d_k$-sector in both pictures when acting on \emph{physical states} and  $\ch_{ij\vert k}^{d_k}$ is isometric to $\ch_{\rm phys}^{d_k}$ \cite{vanrietvelde2020change,Vanrietvelde:2018dit,hohn2020switch,Hoehn:2018whn,hoehn2019trinity,HLSrelativistic,periodtrinity}. In addition, as we shall see, the $\{\mathcal{H}_{ij \vert k}^{d_{k}}\}_{d_k}$ are isometric for different values of $d_k$. The choice of the degeneracy sector in which to carry out the reduction is then arbitrary, as they are in some sense equivalent. One must, however, keep track of the chosen degeneracy sector(s) in order to maintain invertibility. For example, the frame change map in the relational Heisenberg picture is now defined by \cite{vanrietvelde2020change,Vanrietvelde:2018dit,hohn2020switch,Hoehn:2018whn,hoehn2019trinity,HLSrelativistic,periodtrinity} 
\begin{equation}
    \Lambda_{k\to i}^{d_k,d_i}\ce\calr_{d_i}\circ\calr^{-1}_{d_k}  ,\label{deglambda}
\end{equation}
where $\calr_{d_k}^{-1}\ce\cn'(\varepsilon_k)\, \left(\ct_{k,\varepsilon}^{d_k}\right)^{-1}\left(\ket{\varepsilon_k,d_k}_k\otimes \mathbf{1}_{ij}\right) $, and $\cn'(\varepsilon_k)$ is a normalization factor such that $\calr^{-1}_{d_k}\ket{\psi^{d_k}_{ij\vert k}}=\ket{\psi^{d_k}_{\rm phys}}$. This is to be understood as appending the tensor factor $\ket{\varepsilon_k,d_k}_k\otimes$ to the reduced state in $\ch_{ij\vert k}^{d_k}$ (which no longer contains any $k$ degrees of freedom) and then performing the kinematical `re-entangling', \textit{i.e.}\ inverse of the disentangling operation. This map is only an invertible isometry between $\ch_{ij\vert k}^{d_k}\big|_{d_i}$ and $\ch_{jk\vert i}^{d_i}\big|_{d_k}$, where $\big|_{d_i}$ denotes restriction to the $d_i$-degeneracy sector in $\ch_{ij\vert k}^{d_k}$ (and likewise for $\big|_{d_k}$), and not when acting on arbitrary kinematical states in $\ch_i\otimes\ch_j$. In other words, we can only invertibly transform states with support in the overlap of the $d_k$- and $d_i$-sectors, \textit{i.e.}\ in $\ch_{\rm phys}^{d_k}\cap\ch_{\rm phys}^{d_i}$ at the perspecive-neutral level, from $k$'s to $i$'s perspective, and \textit{vice versa} \cite{hohn2020switch,Hoehn:2018whn,HLSrelativistic,periodtrinity}.

\subsection{Commuting subalgebras of relational observables}

Suppose that the reduced physical Hilbert space admits a tensor factorization across the subsystems $i,j$ from $k$'s perspective, \textit{i.e.}\ $\ch_{ij\vert k}^{d_k}\simeq\ch_{i\vert k}^{d_k} \otimes \ch_{j\vert k}^{d_k}$. 
We denote the algebra generated by local subsystem observables on $\ch_{ij|k}^{d_k}$ by $\ca_{ij|k}^{d_k}:=\ca_{i\vert k}^{d_k}\otimes\ca_{j\vert k}^{d_k}$, where $\ca_{i\vert k}^{d_k}\ce\cb(\ch_{i\vert k}^{d_k})$ (hence a type I factor). All statements in the main body pertaining to observable algebras remain true superselection sector-wise in the presence of degeneracies;  hence, in the corresponding statements the algebras should be equipped with the superselection sector label $d_k$. 

Before proving Theorem~\ref{lem_1}, we  recall a few properties of relational observables. The algebra of relational observables of $i$ relative to the QRF $k$ is given by $\ca_{\rm phys}^{i\vert k,d_k}=\calr_{d_k}^{-1}\left(\ca_{i\vert k}^{d_k}\otimes\mathbf{1}_j\right)\,\calr_{d_k}$. Its elements are thus of the form $\calr_{d_k}^{-1}\left(\hat f_i\otimes\mathbf{1}_j\right)\calr_{d_k}$ for some $\hat f_i\in\ca_{i\vert k}^{d_k}$. As shown in \cite{hoehn2019trinity,HLSrelativistic,periodtrinity}, given our assumptions on the constraint $\hat C$, the embedded observables can be written as,\footnote{
In \cite{hoehn2019trinity,HLSrelativistic,periodtrinity}, the relational observables are given in slightly different form as $\hat F^{d_k}_{f_i\otimes 1_j,k}(g)$ and correspond to the value of $\hat f_i$ when the QRF $k$ is in `orientation' $g$. These are equal to the expression here if one replaces $\hat f_i$ by the Heisenberg observable $\hat f_i(g)=e^{ig\hat C_i}\,\hat f_i\,e^{-ig\hat C_i}$, \textit{i.e.}\ $\hat F^{d_k}_{f_i\otimes 1_j,k}(g)\equiv\hat F^{d_k}_{f_i(g)\otimes 1_j,k}$. Here, we drop the frame orientation label $g$ for simplicity. }
\ba
\calr_{d_k}^{-1}\left(\hat f_i\otimes\mathbf{1}_j\right)\calr_{d_k}\,\ket{\psi_{\rm phys}}=\hat F^{d_k}_{f_i\otimes\mathbf{1}_j,k}\,\ket{\psi_{\rm phys}},
\ea
$\forall\,\ket{\psi_{\rm phys}}\in\ch_{\rm phys}$, where 
\ba
\hat F^{d_k}_{f_i\otimes\mathbf{1}_j,k}&=&\int_{G_k}E^{d_k}_k(dg)\otimes\sum_{n=0}^{\infty}\f{i^n}{n!}g^n\big[\hat f_i\otimes\mathbf{1}_j,\hat C_i\otimes\mathbf{1}_j+\mathbf{1}_i\otimes\hat C_j\big]_n\label{relobs}\\
&=&\sum_{n=0}^{\infty}\f{i^n}{n!}\,\hat O^{(n)}_{d_k}\otimes \big[\hat f_i\otimes\mathbf{1}_j,\hat C_i\otimes\mathbf{1}_j+\mathbf{1}_i\otimes\hat C_j\big]_n,\nn
\ea
where $[\hat A,\hat B]_n=\big[\big[\hat A,\hat B\big]_{n-1},\hat B\big]$ is the $n^{\rm th}$-order nested commutator with the convention $\big[\hat A,\hat B\big]_0=\hat A$, $G_k$ is the group generated by $\hat C_k$, and $E_k^{d_k}(dg)={\mu}\,dg\ket{g,d_k}\!\bra{g,d_k}$ and $\hat O^{(n)}_{d_k}$ are the effect density and $n^{\rm th}$-QRF-moment of the $d_k$-sector. This expression can also be written in terms of a so-called $G$-twirl, \textit{i.e.}\ an incoherent group averaging, over $G_k$ \cite{hoehn2019trinity,HLSrelativistic,periodtrinity}. However, the $G$-twirl form will look slightly differently in the $G_k=\rm{U}(1)$ and $G_k=\mathbb{R}$ cases which is why here we work with the power series which is valid for both. These observables are indeed invariant (Dirac) observables, satisfying  \cite{hoehn2019trinity,HLSrelativistic,periodtrinity}
\ba
\big[\hat F^{d_k}_{f_i\otimes\mathbf{1}_j,k},\hat C\big]\ket{\psi_{\rm phys}}=0, \q\q\forall\,\ket{\psi_{\rm phys}}\in\ch_{\rm phys}.\label{diracobs} 
\ea 
In particular, these expressions constitute a quantization of the systematic power-series expansion of classical relational observables developed in \cite{dittrich_partial_2006,dittrich_partial_2007,Thiemann:2004wk}, here for the case of a single constraint (see also \cite{Chataignier:2019kof,Chataignier:2020fys}).

By the arguments in the main body, the assumption on the tensor factorization of the reduced physical Hilbert spaces from the beginning of this section implies that we also have, for example, the two tensor factorizations $\ch_{\rm phys}^{d_A}=\ch_{\rm phys}^{B\vert A,d_A}\otimes\ch_{\rm phys}^{C\vert A,d_A}$ and $\ch_{\rm phys}^{d_B}=\ch_{\rm phys}^{A\vert B,d_B}\otimes\ch_{\rm phys}^{C\vert B,d_B}$ in the superselection sectors of $A$ and $B$. The following result shows that these constitute \emph{different} physical tensor factorizations, including in the overlap of the $d_A$- and $d_B$-superselection sectors, $\ch_{\rm phys}^{d_A}\cap\ch_{\rm phys}^{d_B}$. The reason we have to focus on overlaps has to do with the observation above that Eq.~\eqref{deglambda} is only invertible for such overlaps. \\

\noindent{\bf Theorem~\ref{lem_1}.} 
\emph{The algebra of relational observables of $C$ relative to $A$, is distinct from that of $C$ relative to $B$, \textit{i.e.}\ $\ca_{\rm phys}^{C\vert A,d_A}\big|_{d_B}\neq\ca_{\rm phys}^{C\vert B,d_B}\big|_{d_A}$, where $\big|_{d_k}$ denotes restriction to the $d_k$-sector.}\\

While the restrictions of these algebras to an overlap of superselection sectors are thus distinct, note that they will be isomorphic when $\ch_{\rm phys}^{C\vert A,d_A}\big|_{d_B}\simeq\ch_{\rm phys}^{C\vert B,d_B}\big|_{d_A}$. This will, however, not always be the case, especially when the number of degeneracy sectors for $A$ is different from $B$ (see \cite{hohn2020switch} for an example).

\begin{proof}
We will use the power series expansion of relational observables to show that $\ca_{\rm phys}^{i\vert k,d_k}$ does not commute with $\ca_{\rm phys}^{k\vert j,d_j}$ from which we then infer the claim. Since
\ba
\hat F^{d_k}_{f_i\otimes\mathbf{1}_j,k}=\hat F^{d_k}_{f_i,k}\otimes\mathbf{1}_j,\label{relobs2}
\ea
where $\hat F^{d_k}_{f_i,k}$ is the same expression as in Eq.~\eqref{relobs}, except with the $j$-tensor factor removed, we also have from Eq.~\eqref{diracobs}
\ba
\big[\hat F^{d_k}_{f_i,k},\mathbf{1}_k\otimes\hat C_i+\hat C_k\otimes\mathbf{1}_i\big]\otimes\mathbf{1}_j\ket{\psi_{\rm phys}}=0.\label{commute}
\ea
 When $\hat f_i$ is a constant of motion, $\big[\hat f_i,\hat C_i\big]=0$, Eqs.~\eqref{relobs} and~\eqref{relobs2} imply $\hat F^{d_k}_{f_i,k}=\mathbf{1}_k\otimes\hat f_i$ because for the zeroth moment, $\hat O^{(0)}_k=\mathbf{1}_k$.
Specifically,  we have $\hat F^{d_k}_{\hat C_i,k}=\mathbf{1}_k\otimes\hat C_i$ and clearly the restriction $\hat C_i\big|_{d_k}$ of $\hat C_i$ to the $d_k$-sector lies in $\ca_{i\vert k}^{d_k}$. As shown in \cite{hoehn2019trinity,HLSrelativistic,periodtrinity}, the map $\hat f_i\mapsto \hat F_{f_i\otimes\mathbf{1}_j,k}^{d_k}$ defines a (weak) algebra homomorphism $\ca_{i\vert k}^{d_k}\rightarrow \ca_{\rm phys}^{d_k}$, which here implies that for all $\ket{\psi_{\rm phys}}\in\ch_{\rm phys}$
 \ba
 \hat F_{[f_i,C_i]\otimes\mathbf{1}_j,k}^{d_k}\,\ket{\psi_{\rm phys}}=\big[\hat F^{d_k}_{f_i\otimes\mathbf{1}_j,k},\hat F^{d_k}_{C_i\otimes\mathbf{1}_j,k}\big]\,\ket{\psi_{\rm phys}}=\Big[\hat F^{d_k}_{f_i\otimes\mathbf{1}_j,k},\mathbf{1}_k\otimes\hat C_i\otimes\mathbf{1}_j\Big]\,\ket{\psi_{\rm phys}},\nn
 \ea 
 where on the left hand side we mean with the commutator $[\hat f_i,\hat C_i]$ its restriction to $\ca_{i\vert k}^{d_k}$. Now $\ca_{i\vert k}^{d_k}$ must contain  operators $\hat f_i$ that are not constants of motion if $\hat C_i\big|_{d_k}$ is non-trivial, \textit{i.e.}\ that satisfy $[\hat f_i,\hat C_i]\neq0$ when restricted to $\ca_{i\vert k}^{d_k}$.
 In particular, if $\hat f_i$ is \emph{not} a constant of motion, then also 
 \ba 
 \Big[\hat F^{d_k}_{f_i\otimes\mathbf{1}_j,k},\mathbf{1}_k\otimes\hat C_i\otimes\mathbf{1}_j\Big]\,\ket{\psi_{\rm phys}}\neq0
 \ea 
 for general physical states because only the  relational observable corresponding to zero satisfies $\hat F^{d_k}_{0\otimes 1_j,k}\ket{\psi_{\rm phys}}=0$, $\forall\,\ket{\psi_{\rm phys}}\in\ch_{\rm phys}$.
  Hence, thanks to Eq.~\eqref{commute} we also have
 \ba
 \big[\hat F^{d_k}_{f_i\otimes\mathbf{1}_j,k},\hat C_k\otimes\mathbf{1}_i\otimes\mathbf{1}_j\big]\ket{\psi_{\rm phys}}\neq0
 \ea
 for general physical states. But $\hat F^{d_k}_{f_i\otimes\mathbf{1}_j,k}\in\ca_{\rm phys}^{i\vert k,d_k}$, while clearly (an appropriate restriction of) $\hat C_k\otimes\mathbf{1}_i\otimes\mathbf{1}_j$ is an element of \emph{both} $\ca_{\rm phys}^{k\vert i,d_i}$ and $\ca_{\rm phys}^{k\vert j,d_j}$. For $k=A$, $i=C$ and $j=B$, we thus find $\big[\ca_{\rm phys}^{C\vert A,d_A},\ca_{\rm phys}^{A\vert B,d_B}\big]\neq0$. Then, since $\big[\ca_{\rm phys}^{C\vert B,d_B},\ca_{\rm phys}^{A\vert B,d_B}\big]=0$, we have that $\ca_{\rm phys}^{C\vert A,d_A}\big|_{d_B}\neq\ca_{\rm phys}^{C\vert B,d_B}\big|_{d_A}$, where $\big|_{d_k}$ denotes restriction to the $d_k$-sector.
\end{proof}

\subsection{Example: QRF dependent subsystem locality and correlations}

Here we illustrate Theorem~\ref{lem_1} and the observation of the main text that subsystem locality and correlations are dependent on the choice of QRF 
in a simple example: three particles on a line subject to global translation invariance $\hat C=\hat p_A+\hat p_B+\hat p_C$~\cite{vanrietvelde2020change,Vanrietvelde:2018dit}. 

Solving the constraint for $p_k$, physical states read $\ket{\psi_{\rm phys}}=\int_{\mathbb{R}^2}dp_i dp_j\psi_{ij\vert k}(p_i,p_j)\ket{-p_i-p_j}_k\otimes\ket{p_i}_i\otimes\ket{p_j}$, where the wave function $\psi_{ij\vert k}(p_i,p_j)\ce\psi_{\rm kin}(p_k=-p_i-p_j,p_i,p_j)$, corresponding to the \emph{same} kinematical state in $\ch_{\rm kin}\simeq L^2(\mathbb{R})_A\otimes L^2(\mathbb{R})_B\otimes L^2(\mathbb{R})_C$, will look different for different choices of $k$ and hence feature different correlations structures for different QRF choices. For example, as shown below, any physical state which from $C$'s perspective is separable in $A$ and $B$, is from $B$'s perspective necessarily entangled across $A$ and $C$. Indeed, the tensor factors of $i,j$ correspond to the translation-invariant canonical pairs $(\hat q_i-\hat q_k,\hat p_i), (\hat q_j-\hat q_k,\hat p_j)$ and these generate the algebra of observables on $\ch_{\rm phys}$. Owing to the Stone-von Neumann theorem, we have the tensor factorization $\ch_{\rm phys}\simeq L^2(\mathbb{R})_i\otimes L^2(\mathbb{R})_j$. Since the canonical pairs will be different for different choices of $k$, so will be this physical tensor factorization. This translates into the internal QRF perspectives: the reduction from Eq.~\eqref{qcm} reads $\calr_k =\sqrt{2\pi}\left(\bra{q_k=0}\otimes\mathbf{1}_{ij}\right)\ct_{k,0}$ with unitary disentangler  $\ct_{k,0}=\exp\left[i\hat q_k(\hat p_i+\hat p_j)\right]$ and has the inverse $\calr_k^{-1}=\ct_{k,0}^\dag\left(\ket{p_k=0}\otimes\mathbf{1}_{ij}\right)$ \cite{vanrietvelde2020change}. The translation-invariant canonical pairs become the standard canonical pairs of $i,j$: $\calr_k(\hat q_i-\hat q_k,\hat p_i)\calr_k^{-1}=(\hat q_i,\hat p_i)$ and similarly for $j$, inducing the tensor factorization $\ch_{ij\vert k}=\ch_i\otimes\ch_j$ with the original kinematical $i,j$ Hilbert spaces $\ch_i,\ch_j=L^2(\mathbb{R})$. However, the QRF transformation maps the canonical pairs from $A$'s to $B$'s perspective as $\Lambda_{A\to B}(\hat q_B,\hat p_B)\Lambda_{A\to B}^{-1}=(-\hat q_A,-\hat p_A-\hat p_C)$ and $\Lambda_{A\to B}(\hat q_C,\hat p_C)\Lambda_{A\to B}^{-1} = (\hat q_C-\hat q_A,\hat p_C)$ \cite{Giacomini:2017zju,vanrietvelde2020change}. That is, the  tensor factorization between $B$ and $C$ from $A$'s perspective transforms into a tensor factorization between combinations of $A$ and $C$ degrees of freedom relative to $B$.

Using a different (but  equivalent \cite{vanrietvelde2020change}) formalism without constraints and  physical states and only the internal perspective states, it was already stated in \cite{Giacomini:2017zju} that in general product states map to entangled states under the above QRF transformations.  Here, we give a full proof of  this observation.

Consider a physical state which from $C$'s perspective is separable in $A$ and $B$, \textit{i.e.}\ $\psi_\mathrm{AB\vert C}(p_\mathrm{A},p_\mathrm{B})=\psi_\mathrm{A\vert C}(p_\mathrm{A})\psi_\mathrm{B\vert C}(p_\mathrm{B})$ for some $\psi_\mathrm{A\vert C},\psi_\mathrm{B\vert C} \in L^2(\mathbb{R})$. Transforming to $B$'s perspective \cite{vanrietvelde2020change,Giacomini:2017zju}, one finds
\begin{equation}
\begin{aligned} \label{eNonSeparable}
    \psi_\mathrm{AC\vert B}(p_\mathrm{A},p_\mathrm{C}) &= \psi_\mathrm{A\vert C}(p_\mathrm{A})\psi_\mathrm{B\vert C}(-p_\mathrm{A}-p_\mathrm{C}) .
\end{aligned}
\end{equation}
This is clearly a product state across the factorization defined by the mixed canonical pairs $(\hat q_A-\hat q_C,\hat p_A)$ and $(-\hat q_C,-\hat p_A-\hat p_C)$ mentioned in the main text. However, we now prove by contradiction that the right hand side is necessarily unequal to 
$\psi_\mathrm{A\vert B}(p_\mathrm{A})\psi_\mathrm{C\vert B}(p_\mathrm{C})$
for any $\psi_\mathrm{A\vert B},\psi_\mathrm{C\vert B}\in L^2(\mathbb{R})$. Any state from $C$'s perspective which is separable across $A$ and $B$ is therefore, from $B$'s perspective, necessarily entangled across $A$ and $C$. Due to the invariance of entanglement under local unitary operations \cite{nielsenquantuminfandcomp2011}, this fact is independent of which observable-representation the wavefunction is expressed in.

Indeed, in order for there to exist $\psi_\mathrm{A\vert B},\psi_\mathrm{C\vert B}\in L^2(\mathbb{R})$ such that $\psi_\mathrm{AC\vert B}(p_\mathrm{A},p_\mathrm{C})= \psi_\mathrm{A\vert B}(p_\mathrm{A})\psi_\mathrm{C\vert B}(p_\mathrm{C})$, then from Eq.~\eqref{eNonSeparable} we must have that $\psi_\mathrm{B\vert C}(p+p')=\psi_\mathrm{B\vert C}(p)\psi_\mathrm{B\vert C}(p')$. This, however, implies that $\psi_\mathrm{B\vert C}(p)\notin L^2(\mathbb{R})$, as we now prove by contradiction. Assume that there exists a continuous wavefunction $\psi_{\mathrm B \vert C} (p)$ satisfying the relation $\psi_{\mathrm B \vert C} (p+p')= \psi_{\mathrm B \vert C} (p)\psi_{\mathrm B \vert C} (p')$. Its modulus, $\abs{\psi_{\mathrm B \vert C}}(p)$ must then satisfy the same relation. Note that if there exists $p_{*} \in \mathbb{R}$ such that $\abs{\psi_{\mathrm B \vert C}}(p_{*})=0$, then $\abs{\psi_{\mathrm B \vert C}}(p) =0~\forall p \in \mathbb{R}$ since $\abs{\psi_{\mathrm B \vert C}}(p_{*} +p) = \abs{\psi_{\mathrm B \vert C}} (p_{*}) \abs{\psi_{\mathrm B \vert C}} (p) =0~ \forall p \in\, \mathbb{R}$.  Assuming the wavefunction is non-zero, we may define an auxiliary function $g(p) := \log \left[ \abs{\psi_{\mathrm B \vert C}}(p)\right]$ which is continuous on $\mathbb{R}$ and obeys Cauchy's functional equation, $g(p+p') = g(p) + g(p')$. It follows from the fact that $g(0+0)=g(0)+g(0)=g(0)$ that $g(0)=0$, from $g(p - p) = g(p) + g(-p) = g(0) =0$ that $g(-p) = -g(p)$, and by induction that $g(n p) = n g(p)~ \forall n \in \mathbb{Z}$. Given $m,n \in \mathbb{Z}$ with $n \neq 0$, we  have $g( n \frac{m}{n}) = n g(\frac{m}{n}) =mg(1)$, where we have used the previous identity twice. This implies that $g(p)= g(1)p~ \forall p \in \mathbb{Q}$. Since the continuous functions $g$ and $p \mapsto g(1) p$ agree on $\mathbb{Q}$, a dense subset of $\mathbb{R}$, then they must be equal on the reals. This means that $g(p)=Kp$ for some $K \in \mathbb{R}$, which then implies that $\abs{\psi_{\mathrm B \vert C}}(p)= e^{Kp}$. This is a contradiction, and thus $\psi_{\mathrm{B \vert C}}(p)$ is not square-integrable.

\subsection{A necessary and sufficient condition for physical factorizability according to the kinematical subsystems}

To incorporate degeneracies, the (possibly improper) projector from the kinematical $i,j$ tensor factors onto the reduced physical Hilbert space from $k$'s perspective given in Eq.~\eqref{ePhysicalProjectorBipartite} of the main text must be replaced by $\Pi_{\sigma_{ij \vert k}}^{d_k}:\mathcal{H}_{i} \otimes \mathcal{H}_{j}\to\ch_{ij|k}^{d_k}$, with
\begin{equation} \label{ePhysicalProjectorBipartiteDeg}
    \Pi_{\sigma_{ij \vert k}}^{d_k} \ce  \sum_{d_i,d_{j}}\:\: \intsum_{\:\:\: c_i ,c_j \vert c_i + c_j \in \sigma_{ij \vert k}} \ket{c_{i}, d_{i}} \bra{ c_{i},d_{i}}_{i} \otimes \ket{c_{j}, d_{j}} \bra{ c_{j},d_{j}}_{j} .
\end{equation}
Note that the form of this projector is independent of $d_{k}$, and consequently so too is the projection of a kinematical state onto $\mathcal{H}_{ij \vert k}^{d_{k}}$. The $\{\mathcal{H}_{ij \vert k}^{d_{k}}\}_{d_k}$ therefore have identical structure (as do the observable algebras thereon).

Recall the definition $\sigma_{i\vert jk} \ce \sigma_{i} \cap \spec \left(-\hat{C}_{j}-\hat{C}_{k} \right)$ from the main text. We now prove that the reduced physical Hilbert space $\ch^{d_k}_{ij|k}$ factorizes into $i$ and $j$ subsystems $\ch^{d_k}_{i|k}$ and $\ch^{d_k}_{j|k}$ if and only if
\begin{equation} \label{eConditionSM} 
\sigma_{ij \vert k} = \mathrm{M} \left( \sigma_{i\vert jk} , \sigma_{j\vert ik} \right),
\end{equation}
where $\mathrm{M} \left( \cdot , \cdot \right)$ denotes Minkowski addition. First note that a necessary and sufficient condition for this factorization is that $\Pi_{\sigma_{ij \vert k}}^{d_k}$ in Eq.~\eqref{ePhysicalProjectorBipartiteDeg} factorizes along the same lines, \textit{i.e.~}that
\begin{equation} \label{ePhysicalProjectorSeparableTilde}
\Pi_{\sigma_{ij \vert k}}^{d_k} = \left( \sum_{d_i} ~~ \intsum_{\:\:\: c_{i} \in \tilde{\sigma}_{i}} \ket{c_{i},d_{i}}\bra{c_{i},d_{i}}_i \right) \otimes \left(  \sum_{d_j} ~~ \intsum_{\:\:\: c_{j} \in \tilde{\sigma}_{j}} \ket{c_{j},d_{j}}\bra{c_{j},d_{j}}_j \right).
\end{equation}
for some subsets $\tilde{\sigma}_{i}\subset\sigma_{i}$ and $\tilde{\sigma}_{j}\subset\sigma_{j}$. From Eq.~\eqref{ePhysicalProjectorBipartiteDeg} we see that for this to hold, we must have that
\begin{equation} \label{eConditionTilde}
\sigma_{ij \vert k} = \mathrm{M} \left(\tilde{\sigma}_i , \tilde{\sigma}_j \right) .
\end{equation}
Now assume that Eq.~\eqref{eConditionTilde} holds, and consider some $\lambda_{i}\in\sigma_{i\vert jk}$. By definition, there must exist a $\lambda_{j}\in\sigma_{j\vert ik}$ such that $\lambda_{i}+\lambda_{j}\in\sigma_{ij\vert k}$, and therefore $\Pi^{d_k}_{\sigma_{ij \vert k}} \ket{\lambda_{i}, d_{i}}\bra{\lambda_{i}, d_{i}}_i \otimes \ket{\lambda_{j}, d_{j}}\bra{\lambda_{j}, d_{j}}_j \neq 0 \,\forall d_{i},d_{j}$. From Eq.~\eqref{ePhysicalProjectorSeparableTilde} we see that this implies $\lambda_{i}\in\tilde{\sigma}_{i}$, and since this holds for any $\lambda_{i}\in\sigma_{i\vert jk}$, we must have that $\tilde{\sigma}_{i}=\sigma_{i\vert jk}$. By symmetry, $\tilde{\sigma}_{j}=\sigma_{j\vert ik}$, thus proving the necessity of Eq.~\eqref{eConditionSM} for the factorization $\ch^{d_k}_{ij|k}\simeq\ch^{d_k}_{i|k}\otimes\ch^{d_k}_{i|k}$. The sufficiency of condition~\eqref{eConditionSM} can be proven by noting that it implies 
\begin{equation}
\Pi_{\sigma_{ij \vert k}}^{d_k} = \left( \sum_{d_i} ~~ \intsum_{\:\:\: c_{i} \in \sigma_{i\vert jk}} \ket{c_{i},d_{i}}\bra{c_{i},d_{i}}_i \right) \otimes \left(  \sum_{d_j} ~~ \intsum_{\:\:\: c_{j} \in \sigma_{j\vert ik}} \ket{c_{j},d_{j}}\bra{c_{j},d_{j}}_j \right) ,
\end{equation}
\textit{i.e.} Eq.~\eqref{ePhysicalProjectorSeparableTilde}, and therefore $\ch^{d_k}_{ij|k}\simeq\ch^{d_k}_{i|k}\otimes\ch^{d_k}_{i|k}$.


\subsection{(Non-)factorizability of the observable algebra in frame perspectives}

In the main body, we discussed that projecting kinematical basis elements $\hat A_i\otimes\hat A_j\in\cb(\ch_i)\otimes\cb(\ch_j)$ with the (possibly improper) projector $\Pi_{\sigma_{ij\vert k}^{d_k}}$ into $k$'s perspective necessarily yields observables for $i$ and $j$ that are not of product form across these subsystems, unless condition~\eqref{eCondition} is fulfilled. Here, we discuss a subtlety that arises in this argument when $\hat A_i\otimes\hat A_j$ is diagonal in the eigenbases of $\hat C_i,\hat C_j$, 
\ba
\hat A_i\otimes\hat A_j=\sum_{d_i,d_{j}}\:\: \intsum_{\:\:\: c_i ,c_j }\,A_i^{d_i}(c_i)\,A_j^{d_j}(c_j)\ket{c_i,d_i}_i\!\bra{c_i,d_i}\otimes\ket{c_j,d_j}_j\!\bra{c_j,d_j}\label{prodform}
\ea 
for some coefficients $A_i^{d_i}(c_j)$ (and likewise for $j$), in which case it will commute with $\Pi_{\sigma_{ij\vert k}}^{d_k}$. As we shall now see, the argument remains valid also in this case.

First note that the projection of such a basis element,
\ba 
\Pi_{\sigma_{ij\vert k}}^{d_k}\left(\hat A_i\otimes\hat A_j\right)=\left(\hat A_i\otimes\hat A_j\right)\Pi_{\sigma_{ij\vert k}}^{d_k}=\sum_{d_i,d_{j}}\:\: \intsum_{\:\:\: c_i ,c_j\vert c_i+c_j\in\sigma_{ij\vert k}} \,A_i^{d_i}(c_i)\,A_j^{d_j}(c_j)\ket{c_i,d_i}_i\!\bra{c_i,d_i}\otimes\ket{c_j,d_j}_j\!\bra{c_j,d_j},\label{nonprodform}
\ea 
\emph{cannot} be written in product form across $i$ and $j$, unless condition \eqref{eCondition} is met because of the condition $c_i+c_j\in\sigma_{ij\vert k}$ in the integral-sum.

At this point, we have to distinguish the case in which 
\begin{itemize}
\item[(i)] $\Pi_{\sigma_{ij\vert k}}^{d_k}$ is a proper projector, $\left(\Pi_{\sigma_{ij\vert k}}^{d_k}\right)^2=\Pi_{\sigma_{ij\vert k}}^{d_k}$ and $\ch_{ij\vert k}^{d_k}$ is a proper subspace of $\ch_i\otimes\ch_j$, from the case in which \item[(ii)] it is an improper projector and applying it twice leads to a divergence, formally $\left(\Pi_{\sigma_{ij\vert k}}^{d_k}\right)^2=\infty\cdot\Pi_{\sigma_{ij\vert k}}^{d_k}$, such that $\ch_{ij\vert k}^{d_k}$ is not a subspace of $\ch_i\otimes\ch_j$, but should rather be understood as a space of densities over $\ch_i\otimes\ch_j$ in a rigged Hilbert space construction.
\end{itemize}
Case (i) arises, unless $\sigma_{ij\vert k}$ is discrete while $\hat C_i+\hat C_j$ has continuous spectrum, in which case (ii) occurs \cite{hoehn2019trinity,HLSrelativistic,periodtrinity}.

In case (i), Eq.~\eqref{nonprodform} constitutes an element of the algebra $\cb(\ch_{ij\vert k}^{d_k})$. 

In case (ii) the integral-sum in Eq.~\eqref{nonprodform} is a discrete sum and  a subtlety arises: the operator in this form is \emph{not} an element of $\cb(\ch_{ij\vert k}^{d_k})$. It \emph{cannot} act on reduced physical states, for otherwise it generates a divergence. Indeed, recall that we can write for any reduced physical state in $k$'s perspective $\ch_{ij\vert k}^{d_k}\ni\ket{\psi_{ij\vert k}^{d_k}}=\Pi_{\sigma_{ij\vert k}}^{d_k}\,\ket{\psi_{ij}}$ for some $\ket{\psi_{ij}}\in\ch_{i}\otimes\ch_j$. Hence, $\Pi_{\sigma_{ij\vert k}}^{d_k}\left(\hat A_i\otimes\hat A_j\right)\,\ket{\psi_{ij\vert k}^{d_k}}=\left(\hat A_i\otimes\hat A_j\right)\left(\Pi_{\sigma_{ij\vert k}}^{d_k}\right)^2\,\ket{\psi_{ij}}$, which diverges due to the square of the improper projector. We can circumvent this issue  by simply using the \emph{kinematical} form of this basis observable when acting on reduced physical states. This is consistent because $\hat A_i\otimes\hat A_j$ leaves $\ch_{ij\vert k}^{d_k}$ invariant, seeing that it commutes with $\Pi_{\sigma_{ij\vert k}}^{d_k}$. Notwithstanding, although the kinematical form of such a basis observable is of product form, it is clear that its action on reduced physical states is determined by the terms appearing in Eq.~\eqref{nonprodform}. As such, the action of this observable on reduced physical states is also \emph{not} of product form across $i$ and $j$, except when condition~\eqref{eCondition} is fulfilled.

This last observation can also be understood via the inner product on $\ch_{ij\vert k}^{d_k}$ which in this case is given by $(\phi_{ij\vert k}^{d_k}\vert \psi_{ij\vert k}^{d_k})\ce\braket{\phi_{ij}\vert\psi_{ij\vert k}^{d_k}}$, where $\braket{\cdot\vert \cdot}$ is the inner product on the kinematical tensor factors $\ch_{i}\otimes\ch_j$ and $\ket{\phi_{ij}}$ is any state in that space which projects under $\Pi_{\sigma_{ij\vert k}}^{d_k}$ to $\ket{\phi_{ij\vert k}^{d_k}}\in\ch_{ij\vert k}^{d_k}$ \cite{periodtrinity}.
In other words, also in the reduced physical inner product one acts with kinematical ``bra"-states on physical ``ket"-states and this, of course, is relevant for the spectral decomposition of an observable in the $\hat C_i,\hat C_j$ eigenstates on $\ch_{ij\vert k}^{d_k}$; we can just use its kinematical decomposition.

In conclusion, for both cases (i) and (ii) the reduced physical observable from $k$'s perspective corresponding to a kinematical basis observable $\hat A_i\otimes\hat A_j$ that commutes with the (possibly improper) projector $\Pi_{\sigma_{ij\vert k}}^{d_k}$ also leads to an action on $\ch_{ij\vert k}^{d_k}$ that is \emph{not} of product form across $i$ and $j$, unless condition~\eqref{eCondition} is satisfied.

\subsection{Kinematically commuting subsystem pairs can reduce to physically non-commuting subsystem pairs}

Here we prove Theorem~\ref{thm_2} of the main body which states that, unless condition~\eqref{eCondition} is satisfied, there will exist kinematical observables $\hat A_i\otimes\mathbf{1}_j$ and $\mathbf{1}_i\otimes\hat A_j$, which obviously commute but whose image under reduction with $\Pi_{\sigma_{ij\vert k}}^{d_k}$ does \emph{not} commute. In other words, \emph{kinematically commuting subalgebras, can map into non-commuting observable sets in} $\cb(\ch_{ij\vert k}^{d_k}$). To avoid the technical subtleties surrounding case (ii) of the previous subsection, we first focus on case (i) in Lemma~\ref{lem_2} and return to (ii) in Lemma~\ref{lLemma3} below. The conjunction of Lemmas~\ref{lem_2} and~\ref{lLemma3} yields the statement of Theorem~\ref{thm_2}.

\begin{Lemma}\label{lem_2}
If $\Pi_{\sigma_{ij\vert k}}^{d_k}$ is a proper projector and condition~\eqref{eCondition} is not satisfied, there will exist $\hat A_i\in\cb(\ch_i)$ and $\hat A_j\in\cb(\ch_j)$ such that $\big[\hat A_{i\vert k},\hat A_{j\vert k}\big]\neq0$, where $\hat A_{i\vert k}\ce\Pi_{\sigma_{ij\vert k}}^{d_k}\,\left(\hat A_i\otimes\mathbf{1}_j\right)\,\Pi_{\sigma_{ij\vert k}}^{d_k}\in\cb(\ch_{ij\vert k}^{d_k})$ and $\hat A_{j\vert k}\ce\Pi_{\sigma_{ij\vert k}}^{d_k}\,\left(\mathbf{1}_j\otimes\hat A_j\right)\,\Pi_{\sigma_{ij\vert k}}^{d_k}\in\cb(\ch_{ij\vert k}^{d_k})$. By contrast, if condition~\eqref{eCondition} is satisfied, then $\big[\hat A_{i\vert k},\hat A_{j\vert k}\big]=0$ for all $\hat A_i\in\cb(\ch_i)$ and $\hat A_j\in\cb(\ch_j)$.
\end{Lemma}

\begin{proof}
It is clear that if at least one of $\hat A_i\otimes\mathbf{1}_j$ or $\mathbf{1}_i\otimes\hat A_j$ commutes with $\Pi_{\sigma_{ij\vert k}}^{d_k}$, then $\big[\hat A_{i\vert k},\hat A_{j\vert k}\big]=0$. We thus consider the situation in which neither commutes with the projector and there certainly exist $\hat A_i\in\cb(\ch_i)$ and $\hat A_j\in\cb(\ch_j)$ such that this is true if $\hat C_i,\hat C_j$ are not just multiples of the identity. We can decompose both observables in the eigenbases of $\hat C_i$ and $\hat C_j$:
\ba
\hat A_i\otimes\mathbf{1}_j&=&\sum_{d_i,d_i'}\:\: \intsum_{\:\:\: c_i ,c_i'}\,A_i(c_i,d_i,c_i',d_i')\ket{c_i,d_i}_i\!\bra{c_i',d_i'}\otimes
\:\: \intsum_{\:\:\: c_j}\sum_{d_j}\ket{c_j,d_j}_j\!\bra{c_j,d_j},\nn\\
\mathbf{1}_i\otimes\hat A_j&=&
\:\: \intsum_{\:\:\: c_i}\sum_{d_i}\ket{c_i,d_i}_i\!\bra{c_i,d_i}\otimes\sum_{d_j,d_j'}\:\: \intsum_{\:\:\: c_j ,c_j'}\,A_j(c_j,d_j,c_j',d_j')\ket{c_j,d_j}_j\!\bra{c_j',d_j'},\nn 
\ea 
for some coefficients $A_i,A_j$. Using Eq.~\eqref{ePhysicalProjectorBipartite}, this yields
\ba
\hat A_{i\vert k}&=&\Pi_{\sigma_{ij\vert k}}^{d_k}\,\left(\hat A_i\otimes\mathbf{1}_j\right)\,\Pi_{\sigma_{ij\vert k}}^{d_k}=\sum_{d_i,d_i',d_j}\:\: \underset{c_i+c_j,c_i'+c_j\in\sigma_{ij\vert k}}{\intsum_{\:\:\: c_i ,c_i',c_j\vert}}\,A_i(c_i,d_i,c_i',d_i')\ket{c_i,d_i}_i\!\bra{c_i',d_i'}\otimes\ket{c_j,d_j}_j\!\bra{c_j,d_j},\nn \\
\hat A_{j\vert k}&=&\Pi_{\sigma_{ij\vert k}}^{d_k}\,\left(\mathbf{1}_i\otimes\hat A_j\right)\,\Pi_{\sigma_{ij\vert k}}^{d_k}=\sum_{d_i,d_j,d_j'}\:\: \underset{c_i+c_j,c_i+c_j'\in\sigma_{ij\vert k}}{\intsum_{\:\:\: c_i ,c_j,c_j'\vert}}\,A_j(c_j,d_j,c_j',d_j')\ket{c_i,d_i}_i\!\bra{c_i,d_i}\otimes\ket{c_j,d_j}_j\!\bra{c'_j,d_j'}.\nn 
\ea 
Hence, in conjunction,
\ba
\big[\hat A_{i\vert k},\hat A_{j\vert k}\big]&=&\sum_{d_i,d_i',d_j,d_j'}\:\: \underset{c_i+c_j,c_i'+c_j,c_i'+c_j'\in\sigma_{ij\vert k}}{\intsum_{\:\:\: c_i ,c_i',c_j,c_j'\vert}}\,A_i(c_i,d_i,c_i',d_i')\,A_j(c_j,d_j,c_j',d_j')\ket{c_i,d_i}_i\!\bra{c_i',d_i'}\otimes\ket{c_j,d_j}_j\!\bra{c_j',d_j'}\nn\\
&&\q -\sum_{d_i,d_i',d_j,d_j'}\:\: \underset{c_i+c_j,c_i+c_j',c_i'+c_j'\in\sigma_{ij\vert k}}{\intsum_{\:\:\: c_i ,c_i',c_j,c_j'\vert}}\,A_i(c_i,d_i,c_i',d_i')\,A_j(c_j,d_j,c_j',d_j')\ket{c_i,d_i}_i\!\bra{c_i',d_i'}\otimes\ket{c_j,d_j}_j\!\bra{c_j',d_j'}.\nn\\\label{aacom}
\ea 
The only difference between the two terms on the right hand side is that the first term includes the summing condition $c_i'+c_j\in\sigma_{ij\vert k}$, which the second term \emph{a priori} does not encompass, while the second term includes the summing condition $c_i+c_j'\in\sigma_{ij\vert k}$, which the first term \emph{a priori} does not contain. 

Now when $\sigma_{ij\vert k}=M(\tilde\sigma_i,\tilde\sigma_j)$, \textit{i.e.}\ condition~\eqref{eCondition} is satisfied, it is clear that, in the first term, the three summing conditions combined also imply the additional summing condition $c_i+c_j'\in\sigma_{ij\vert k}$ of the second term. Likewise, in this case the summing condition of the first term $c_i'+c_j\in\sigma_{ij\vert k}$ is implied by the three summing conditions of the second term. Thus, when condition~\eqref{eCondition} is met, the commutator vanishes. 

By contrast, when condition~\eqref{eCondition} is not fulfilled, \textit{i.e.}\ $\sigma_{ij\vert k}\neq M(\tilde\sigma_i,\tilde\sigma_j)$, then, in the first term, the sum will include terms with  $c_i,c_j'$ such that $c_i+c_j'\notin\sigma_{ij\vert k}$ and similarly, the sum of the second term will include terms with $c_i',c_j$ such that $c_i'+c_j\notin\sigma_{ij\vert k}$. Note that these are always conditions on pairs of variables such that one of them is contained in the coefficient function $A_i$ and the other in $A_j$, both of which can be chosen independently as they correspond to $\hat A_i\in\cb(\ch_i)$ and $\hat A_j\in\cb(\ch_j)$, respectively. It is clear that $A_i$ and $A_j$ can be chosen such that $A_i(c_i,d_i,c_i',d_i')A_j(c_j,d_j,c_j',d_j')\neq0$ when $c_i+c_j'\notin\sigma_{ij\vert k}$ and/or $c_i'+c_j\notin\sigma_{ij\vert k}$ and in this case, the commutator does \emph{not} vanish. This proves the claim.
\end{proof}

Next, we also discuss the more subtle case (ii). In this case, we cannot act with the $\hat A_{i\vert k},\hat A_{j\vert k}$ on reduced physical states as this generates a divergence. We thus have to take the improper nature of the projector into account and redefine $\hat A_{i\vert k},\hat A_{j\vert k}$.

\begin{Lemma} \label{lLemma3}
If $\Pi_{\sigma_{ij\vert k}}^{d_k}$ is an \emph{improper} projector and condition~\eqref{eCondition} is not satisfied, there will exist $\hat A_i\in\cb(\ch_i)$, $\hat A_j\in\cb(\ch_j)$ and $\ket{\psi_{ij\vert k}^{d_k}}\in\ch_{ij\vert k}^{d_k}$ such that $\big[\hat A'_{i\vert k},\hat A'_{j\vert k}\big]\,\ket{\psi_{ij\vert k}^{d_k}}\neq0$, where $\hat A'_{i\vert k}\ce\Pi_{\sigma_{ij\vert k}}^{d_k}\,\left(\hat A_i\otimes\mathbf{1}_j\right)\in\cb(\ch_{ij\vert k}^{d_k})$ and $\hat A'_{j\vert k}\ce\Pi_{\sigma_{ij\vert k}}^{d_k}\,\left(\mathbf{1}_j\otimes\hat A_j\right)\in\cb(\ch_{ij\vert k}^{d_k})$. Note that this requires $\big[\Pi_{\sigma_{ij\vert k}}^{d_k},\hat A_i\otimes\mathbf{1}_j\big]\neq0$ and similarly for $\hat A_j$. By contrast, if condition~\eqref{eCondition} is fulfilled, $\big[\hat A'_{i\vert k},\hat A'_{j\vert k}\big]\,\ket{\psi_{ij\vert k}^{d_k}}=0$ for all $\ket{\psi_{ij\vert k}^{d_k}}\in\ch_{ij\vert k}^{d_k}$ for such $\hat A'_{i\vert k},\hat A'_{j\vert k}$.
\end{Lemma}

\begin{proof}
On account of the improper nature of the projector, it follows that that $\big[\Pi_{\sigma_{ij\vert k}}^{d_k},\hat A_i\otimes\mathbf{1}_j\big]\neq0$ for $\hat A'_{i\vert k}$ defined in the statement of the Lemma to be contained in $\cb(\ch_{ij\vert k}^{d_k})$ and likewise for $\hat A'_{j\vert k}$. The statement then follows by recalling that $\ket{\psi_{ij\vert k}^{d_k}}=\Pi_{\sigma_{ij\vert k}}^{d_k}\,\ket{\psi_{ij}}$ for some $\ket{\psi_{ij}}\in\ch_i\otimes\ch_j$, noting that
\ba 
\big[\hat A'_{i\vert k},\hat A'_{j\vert k}\big]\,\ket{\psi_{ij\vert k}^{d_k}} = \big[\hat A_{i\vert k},\hat A_{j\vert k}\big]\,\ket{\psi_{ij}},\nn
\ea
where $\big[\hat A_{i\vert k},\hat A_{j\vert k}\big]$ coincides with the expression given in Eq.~\eqref{aacom}, and the proof of Lemma~\ref{lem_2}.
\end{proof}
In conjunction, this proves Theorem~\ref{thm_2}.

\subsection{Two kinematical qubits can be one physical qutrit}

As an example of factorizability condition of Equation~\eqref{eCondition}, consider a system consisting of two identical qubits and a free particle subject to the Hamiltonian constraint
\begin{equation}
\hat{C} = -\frac{\hat{\sigma}^{(z)}_{A}}{2} -\frac{\hat{\sigma}^{(z)}_{B}}{2} + \frac{\hat{p}_{C}^{2}}{2m} 
\end{equation} 
where $\hat{\sigma}^{(z)}_{j}:=\ket{0}\!\bra{0}_j -\ket{1}\!\bra{1}_j$ is the Pauli $z$ matrix corresponding to system $j$. Choosing $C$ as a temporal reference frame (\textit{i.e.}\ a clock), we have $\sigma_{AB\vert C} = \lbrace 0, -1 \rbrace$. There are then no subsets of $\sigma_{A}$ and $\sigma_{B}$ satisfying the factorization condition, Eq.~\eqref{eCondition}, thus $\ch_{AB|C}^{d_C}$,  where $d_C=\pm1$, does not factorize into a subset of $\mathcal{H}_A$ and a subset of $\mathcal{H}_B$, though we can combine certain elements of $\mathcal{H}_A$ and $\mathcal{H}_B$ to label the states of $\ch_{AB|C}^{d_C}$.

The kinematical space of the two qubits, $\mathcal{H}_A \otimes \mathcal{H}_B$, is $4$-dimensional. On the other hand, the reduced physical Hilbert space corresponding to conditioning on clock $C$, \textit{i.e.}\ $\Pi^{d_C}_{\sigma_{AB \vert C}} ( \mathcal{H}_A \otimes \mathcal{H}_B )$,  is three dimensional, with a basis constructed from the kinematic states $\lbrace \ket{0_{A} 1_{B}}, \ket{1_{A} 0_{B}}, \ket{0_{A} 0_{B}} \rbrace$. The first two states correspond to $\hat{C}_C$ taking the value $0$, and the third state corresponds to $\hat{C}_C$ taking the value $1$. In this reference frame, the system evolves according to the Hamiltonian~\cite{hoehn2019trinity} $\hat{H}_{AB}^\mathrm{phys} :=\Pi_{\sigma_{AB \vert C}}^{d_C} (\hat{C}_A +\hat{C}_B)$, whose three eigenstates correspond to eigenvalues $\lbrace 0, 0, -1 \rbrace$, and does not separate into two subsystems with additive energies. The system thus evolves not as two qubits, but rather as a single qutrit.

\subsection{Frame-dependent factorizability}

In the main text, we considered constraints of the form where $\sigma_\mathrm{A}=\mathds{R}_{+}$, $\sigma_\mathrm{B} =\mathds{R}_{+}$, and $\sigma_\mathrm{C} =\mathds{R}$. From $C$'s perspective, the kinematical factorizability of $A$ and $B$ is unchanged, \textit{i.e.}\ $\Pi_{\sigma_{\mathrm{AB} \vert \mathrm{C}}} ( \mathcal{H}_\mathrm{A} \otimes \mathcal{H}_\mathrm{B} ) = \mathcal{H}_\mathrm{A} \otimes \mathcal{H}_\mathrm{B}$. From $B$'s perspective, on the other hand, the reduced physical Hilbert space does not factor into $A$ and $C$ parts, as we now prove by contradiction. First assume that condition~\eqref{eCondition} is satisfied, namely that there exists $\tilde{\sigma}_\mathrm{A} \subseteq \sigma_\mathrm{A}$ and $\tilde{\sigma}_\mathrm{C}\subseteq \sigma_\mathrm{C}$ such that $\sigma_{\mathrm{AC} \vert \mathrm{B}} = \mathrm{M} \left(\tilde{\sigma}_\mathrm{A} , \tilde{\sigma}_\mathrm{C} \right)$. Note that $\sigma_{\mathrm{AC} \vert \mathrm{B}}=\mathds{R}_{-}$ (\textit{i.e.}\ the negative reals) for the class of constraints considered. Now consider some positive eigenvalue $a$ of $\hat{C}_\mathrm{A}$ and the corresponding negative eigenvalue $-a$ of $\hat{C}_\mathrm{C}$. Their sum, $0$, is in $\sigma_{\mathrm{AC} \vert \mathrm{B}}$, and thus $a \in \tilde{\sigma}_\mathrm{A}$ and $-a\in \tilde{\sigma}_\mathrm{C}$. Then consider a smaller positive eigenvalue $a'<a$ of $\hat{C}_\mathrm{A}$, and the corresponding negative eigenvalue $-a'$ of $\hat{C}_\mathrm{C}$. Their sum, again $0$, is in $\sigma_{\mathrm{AC} \vert \mathrm{B}}$, and thus $a' \in \tilde{\sigma}_\mathrm{A}$ and $-a'\in \tilde{\sigma}_\mathrm{C}$. Finally, consider the positive eigenvalue $a$ of $\hat{C}_\mathrm{A}$, and the negative eigenvalue $-a'$ of $\hat{C}_\mathrm{C}$. Now their sum is greater than zero, and thus not an element of $\sigma_{\mathrm{AC} \vert \mathrm{B}}$. Yet their sum is by definition an element of $\mathrm{M} \left(\tilde{\sigma}_\mathrm{A} , \tilde{\sigma}_\mathrm{C} \right)$, and by hypothesis $\sigma_{\mathrm{AC} \vert \mathrm{B}} = \mathrm{M} \left(\tilde{\sigma}_\mathrm{A} , \tilde{\sigma}_\mathrm{C} \right)$. We thus arrive at a contradiction, and there therefore does not exist $\tilde{\sigma}_\mathrm{A} \subseteq \sigma_\mathrm{A}$ and $\tilde{\sigma}_\mathrm{C}\subseteq \sigma_\mathrm{C}$ such that $\sigma_{\mathrm{AC} \vert \mathrm{B}} = \mathrm{M} \left(\tilde{\sigma}_\mathrm{A} , \tilde{\sigma}_\mathrm{C} \right)$, and consequently, taking possible degeneracies of $\hat C_B$ into account,
\begin{equation}\label{nofac}
\Pi^{d_B}_{\sigma_{\mathrm{AC} \vert \mathrm{B}}} \left( \mathcal{H}_\mathrm{A} \otimes \mathcal{H}_\mathrm{C} \right) \neq \tilde{\mathcal{H}}^{d_B}_\mathrm{A} \otimes \tilde{\mathcal{H}}^{d_B}_\mathrm{C} 
\end{equation}
for some $\tilde{\mathcal{H}}^{d_B}_\mathrm{A}\subseteq {\mathcal{H}}_\mathrm{A}$ and $\tilde{\mathcal{H}}^{d_B}_\mathrm{C}\subseteq {\mathcal{H}}_\mathrm{C}$. Therefore, the kinematical factorizability into an $\mathrm{A}$ part and a $\mathrm{C}$ part does not survive on the reduced physical space in $B$'s perspective.

As a specific example of this type of constraint, we considered
\begin{equation}
\hat{C} = \frac{\hat{p}_{A}^{2}}{2} + \frac{\hat{p}_{B}^{2}}{2} + \hat{p}_{C} .\label{examplesup}
\end{equation}
To demonstrate that the distinction between kinematical subsystems carries over to $\ch_{AB|C}$, but not to $\ch_{AC|B}^{d_B}$, where $d_B=\pm1$, we considered the kinematical canonical pairs ${(\hat{x}_{i},\hat{p}_{i}),}(\hat{x}_{j},\hat{p}_{j})$ on {$\ch_i$ and} $\mathcal{H}_{j}$. From $C$'s perspective, we have $i,j=A,B$ and these appear {unaltered} as
\begin{align}
   {\Pi_{\sigma_{AB\vert C}}}\, (\hat{x}_{A}\otimes\mathbf{1}_{B},\hat{p}_{A}\otimes\mathbf{1}_{B})\, {\Pi_{\sigma_{AB\vert C}}} &=
  (\hat{x}_{A}\otimes\mathbf{1}_{B},\hat{p}_{A}\otimes\mathbf{1}_{B}) \\
   {\Pi_{\sigma_{AB\vert C}}}\,(\mathbf{1}_{A}\otimes\hat{x}_{B},\mathbf{1}_{A}\otimes\hat{p}_{B})\,{\Pi_{\sigma_{AB\vert C}}}
   &=
   (\mathbf{1}_{A}\otimes\hat{x}_B,\mathbf{1}_{A}\otimes\hat{p}_{B}) ,
\end{align}
since in this case $\Pi_{\sigma_{AB\vert C}}=\mathbf{1}_{AB}$. By contrast, from $B$'s perspective, we have {$i,j=A,C$ and} would like to compute 
\begin{align}
    (\hat{x}_{A\vert B},\hat{p}_{A\vert B})&\ce{\Pi^{d_B}_{\sigma_{AC\vert B}}}\,(\hat{x}_{A}\otimes\mathbf{1}_{C},\hat{p}_{A}\otimes\mathbf{1}_{C})\,{\Pi^{d_B}_{\sigma_{AC\vert B}}}
     \\
  (\hat{x}_{C\vert B},\hat{p}_{C\vert B})&\ce{\Pi^{d_B}_{\sigma_{AC\vert B}}}\,   (\mathbf{1}_{A}\otimes\hat{x}_{C},\mathbf{1}_{A}\otimes\hat{p}_{C})\,{\Pi^{d_B}_{\sigma_{AC\vert B}}} .
\end{align}
$\Pi_{\sigma_{AC\vert B}}^{d_B}$ is a proper projector in this case. In order to compute these expressions, first note that due to the degeneracy of the operator $\hat p^2$, a complete eigenbasis for the single particle Hilbert space is given by $\ket{\fp^2, d}$, where $\fp^2 \in \Rl$, $d = \pm 1$. This basis is related to the momentum eigenbasis by $\ket p = \ket{\fp^2, d}$ with $d = {\rm sign}(p) 1$  and permits us to write the identity
\begin{align}
\hat p = \int p \ketbra{p}{p} dp = \sum_{d = \pm 1} \int_{\fp^2 = 0}^\infty d \sqrt{\fp^2} \ketbra{\fp^2,d}{\fp^2,d} d\fp^2
\end{align}
in a form which is useful in the calculation. For the constraint in Eq.~\eqref{examplesup}, this yields
\begin{align}
  \hat x_{A|B} &=  \sum_{d_A,d_A' = \pm 1}  \int_{\fp_A^2 = 0}^{\infty} \int_{(\fp_A^2)' = 0}^{\infty} \int_{p_C = -\infty}^{{\rm min}(-\fp_A^2, -(\fp_A^2)')/2 } \braket{p_A|\hat x |p_A'}  \ketbra{\fp_A^2,d_A}{(\fp_A^2)',d_A'} \otimes \ketbra{p_C}{p_C} d\fp_A^2 dp_C   d(\fp_A^2)' ,\\
   \hat p_{A|B} &= \sum_{d_A = \pm 1} \int_{\fp_A^2 = 0}^\infty \int_{p_C = - \infty}^{p_C = -\fp_A^2/2} d_A \sqrt{\fp_A^2} \ketbra{\fp_A^2,d_A}{\fp_A^2,d_A} \otimes \ketbra{p_C}{p_C} d\fp_A^2 dp_C,\\
      \hat x_{C|B} &= \sum_{d_A = \pm 1} \int_{\fp_A^2 = 0}^\infty \int_{p_C = - \infty}^{p_C = -\fp_A^2/2} \int_{p_C' = - \infty}^{p_C' = -\fp_A^2/2} \ketbra{\fp_A^2,d_A}{\fp_A^2,d_A} \otimes \braket{p_C|\hat x |p_C'}  \ketbra{p_C}{p_C'} d\fp_A^2 dp_C dp_C',\\
   \hat p_{C|B} &=  \sum_{d_A = \pm 1} \int_{\fp_A^2 = 0}^\infty \int_{p_C = - \infty}^{p_C = -\fp_A^2/2} \ketbra{\fp_A^2,d_A}{\fp_A^2,d_A} \otimes p_C \ketbra{p_C}{p_C} d\fp_A^2 dp_C.
\end{align}

As one can check, in accordance with Lemma~\ref{lem_2}, the canonical commutation relations between the kinematical observables $\hat{x}_{A}$, $\hat{p}_{A}$, $\hat{x}_{C}$, and $\hat{p}_{C}$ are not preserved for their corresponding physical observables on $\ch_{AC|B}^{d_B}$. While all brackets involving $\hat p_{A\vert B},\hat p_{C\vert B}$ are left invariant, $[\hat p_{A|B},\hat p_{C|B}] =[\hat x_{A|B},\hat p_{C|B}] = [\hat p_{A|B},\hat x_{C|B}] = 0$, since $\hat{p}_{A}$, $\hat{p}_{C}$ commute with the constraint, and therefore with $\Pi_{AC\vert B}^{d_B}$, one now finds that $[\hat x_{A\vert B},\hat x_{C\vert B}]\neq0$.  
Indeed,
\begin{align}
\hat x_{A|B}\,\hat x_{C|B}  &=\sum_{d_A,d_A' = \pm 1}  \int_{\fp_A^2 = 0}^{\infty} \int_{(\fp_A^2)' = 0}^{\infty} \int_{p_C = -\infty}^{p_C = {\rm min}(-\fp_A^2, -(\fp_A^2)')/2}  \int_{p_C' = - \infty}^{p_C' = -(\fp_A^2)'/2} \braket{p_A|\hat x |p_A'}  \ketbra{\fp_A^2,d_A}{(\fp_A^2)',d_A'} \nn\\
&\q\q\q\q\q\q\q\q\q\q\q\q\q\q\q\q\q\q\q\q\q\q\q\otimes \braket{p_C|\hat x |p_C'}  \ketbra{p_C}{p_C'}\, d\fp_A^2   d(\fp_A^2)'\, dp_C \,dp_C',\\
\hat x_{C|B}\,\hat x_{A|B}  &=\sum_{d_A,d_A' = \pm 1}  \int_{\fp_A^2 = 0}^{\infty} \int_{(\fp_A^2)' = 0}^{\infty}\int_{p_C = - \infty}^{p_C = -\fp_A^2/2} \int_{p_C' = -\infty}^{{\rm min}(-\fp_A^2, -(\fp_A^2)')/2 }   \braket{p_A|\hat x |p_A'}  \ketbra{\fp_A^2,d_A}{(\fp_A^2)',d_A'} \nn\\
&\q\q\q\q\q\q\q\q\q\q\q\q\q\q\q\q\q\q\q\q\q\q\q\otimes \braket{p_C|\hat x |p_C'}  \ketbra{p_C}{p_C'}\, d\fp_A^2   d(\fp_A^2)' \,dp_C \,dp_C'.
\end{align}
$[\hat x_{A|B},\hat x_{C|B}]$ will not be zero  because both terms contain factor elements $\ketbra{p_C}{p_C'}$ but the values of $p_C$ and $p_C'$ are not identical for both terms. For instance consider the case where $\fp_A^2 = 2$ and $(\fp_A^2)' = 4$. Then the term   $\hat x_{C|B}\,\hat x_{A|B}$ will contain an element $\ketbra{p_C= -1}{p_C' = -2}$ whereas the term $\hat x_{A|B}\,\hat x_{C|B}$ will not, since for this term $\ketbra{p_C}{p_C'}$ is such that the value of $p_C$ is bounded above by $-2$.

Hence, the transformed canonical pairs no longer commute with one another and as such $(\hat x_{A\vert B},\hat p_{A\vert B}),(\hat x_{C\vert B},\hat p_{C\vert B})$ do \emph{not} induce a tensor factorization across $A$ and $C$ in the reduced physical Hilbert space $\ch_{AC\vert B}^{d_B}$ for both $d_B=\pm1$, thereby elucidating Eq.~\eqref{nofac}.

As mentioned in the main body, this does not imply, however, that $\ch_{AC\vert B}^{d_B}$ does not admit a tensor factorization. Indeed, it is simple to construct one by translating the canonical pairs $(\hat x_A\otimes\mathbf{1}_B,\hat p_A\otimes\mathbf{1}_B),(\mathbf{1}_A\otimes\hat x_B,\mathbf{1}_A\otimes\hat p_B)$ from $C$'s perspective via the QRF transformation $\Lambda_{C\to B}^{d_B}$ into $B$'s perspective. The latter, being an invertible isometry when acting on reduced physical states, will preserve commutation relations. Hence, the images of the canonical pairs from $C$'s perspective will again be mutually commuting canonical pairs in $B$'s perspective and thereby induce a corresponding tensor factorization of $\ch_{AC\vert B}^{d_B}$. However, these new canonical pairs in $B$'s perspective will be combinations of $A$ and $C$ degrees of freedom; for example, using the maps constructed in \cite{hohn2020switch,HLSrelativistic} one can check that $\Lambda_{C\to B}^{d_B}(\hat x_A\otimes\mathbf{1}_B,\hat p_A\otimes\mathbf{1}_B)\,\Lambda_{B\to C}^{d_B} = \Pi_{\sigma_{AC\vert B}}^{d_B}\,(\hat x_A\otimes\mathbf{1}_C-2\hat p_A\otimes\hat x_C,\hat p_A\otimes\mathbf{1}_C)$.\footnote{Note that both operators actually commute with $\Pi_{\sigma_{AC\vert B}}^{d_B}$ since they commute with $\hat C_A\otimes\mathbf{1}_C+\mathbf{1}_A\otimes\hat C_C$.} Consequently, the  thus induced tensor factorization of $\ch_{AC\vert B}^{d_B}$ is \emph{not} one between $A$ and $C$.

All of this can also be understood at the level of the physical Hilbert space $
\ch_{
\rm phys}$. For example, using the reduction map (and its inverse) one can embed the canonical pairs $(\hat x_A\otimes\mathbf{1}_B,\hat p_A\otimes\mathbf{1}_B),(\mathbf{1}_A\otimes\hat x_B,\mathbf{1}_A\otimes\hat p_B)$ from $C$'s perspective into the perspective-neutral setting as the relational observables explicitly given in \cite{hohn2020switch} (here corresponding to clock $C$ reading $\tau=0$) 
\ba 
\hat F_{x_A\otimes 1_B,C}&=&\hat x_A\otimes\mathbf{1}_B\otimes\mathbf{1}_C-2\,\hat p_A\otimes\mathbf{1}_B\otimes\hat x_C,\nn\\
\hat F_{p_A\otimes 1_B,C}&=&\hat p_A\otimes\mathbf{1}_B\otimes\mathbf{1}_C,\nn\\
\hat F_{1_A\otimes x_B,C}&=&\mathbf{1}_A\otimes\hat x_B\otimes\mathbf{1}_C-2\,\mathbf{1}_A\otimes\hat p_B\otimes\hat x_C,\nn\\
\hat F_{1_A\otimes p_B,C}&=&\mathbf{1}_A\otimes\hat p_B\otimes\mathbf{1}_C.\nn
\ea 
Clearly, these are mutually commuting canonically conjugate pairs and thereby induce a corresponding physical tensor factorization of $\ch_{\rm phys}$ via the Stone-von Neumann theorem. 

By contrast, the modified canonical pairs $(\hat x_{A\vert B},\hat p_{A\vert B}),(\hat x_{C\vert B},\hat p_{C\vert B})\in\cb(\ch_{AC\vert B}^{d_B})$ correspond to the relational observables $\hat F^{d_B}_{x_A\otimes1_C,B},\hat F^{d_B}_{p_A\otimes1_C,B}$ and $
\hat F^{d_B}_{1_A\otimes x_C,B},\hat F_{1_A\otimes p_C,B}^{d_B}$ (here corresponding to clock $B$ reading $\tau=0$) \cite{HLSrelativistic}. These actually define equivalence relations on physical states
\ba 
\hat F^{d_B}_{x_A\otimes1_C,B}\ket{\psi_{\rm phys}}=\hat F^{d_B}_{x_{A\vert B},B}\ket{\psi_{\rm phys}},\nn
\ea
(likewise for the other canonical variables) and satisfy a homomorphism in the form \cite{HLSrelativistic}
\ba 
\hat F^{d_B}_{[f,g],B}\ket{\psi_{\rm phys}}=\Big[\hat F^{d_B}_{f,B},\hat F^{d_B}_{g,B}\Big]\ket{\psi_{\rm phys}},
\ea
for any $\hat f,\hat g\in\cb(\ch_{AC\vert B}^{d_B})$ and $\ket{\psi_{\rm phys}}\in\ch_{\rm phys}$. Hence, since as we have shown $[\hat x_{A\vert B},\hat x_{C\vert B}]\neq0$, already these relational Dirac observables form mutually non-commuting canonical pairs on $\ch_{\rm phys}$. Accordingly, the relational observables relative to $B$ corresponding to the kinematical operators $(\hat x_A,\hat p_A),(\hat x_C,\hat p_C)$ also do \emph{not} induce a physical tensor factorization of $\ch_{\rm phys}$ in terms of these observables.

\end{document}